\documentclass[final,journal]{IEEEtran}

\usepackage[
            left=0.75in,right=0.75in,top=1in,bottom=0.75in,%
            footskip=.25in]{geometry}

\usepackage{mathtools,amssymb,graphicx,times,comment}
\usepackage{amsmath,amsthm}
\usepackage{bbm}
  {
      \newtheorem{assumption}{Assumption}
  }

\usepackage{tcolorbox}
\usepackage[hidelinks]{hyperref}
\definecolor{darkred}{RGB}{150,0,0}
\definecolor{darkgreen}{RGB}{0,150,0}
\definecolor{darkblue}{RGB}{0,0,150}
\hypersetup{colorlinks=true, linkcolor=red, citecolor=darkgreen, urlcolor=darkblue}

\newcommand{\deltac}{\delta^c}

\newcommand{\maxx}{\vee}

\newcommand{\psin}[1]{\|#1\|_{\psi_2}}

\newcommand{\TcK}{\Tc_{\Kc,\x_0}}

\newcommand{\bal}{\begin{align}}
\newcommand{\eal}{\end{align}}

\DeclarePairedDelimiterX{\inp}[2]{\langle}{\rangle}{#1, #2}

\newcommand{\Id}{\mathbf{I}}

\newcommand{\wg}{\omega}

\newcommand{\ksi}{\xi}

\newcommand{\wh}{{\widehat\w}}

\usepackage{dsfont}
\newcommand{\ind}[1]{{\mathds{1}}_{\{#1\}}}

\newcommand{\Pro}{\mathbb{P}}

\theoremstyle{theorem}
\newtheorem{propo}{Proposition}[section]
\newtheorem{thm}{Theorem}[section]
\newtheorem{lem}{Lemma}[section]
\theoremstyle{remark}
\newtheorem{remark}{Remark}
\theoremstyle{definition}
\newtheorem{defn}{Definition}[section]

\DeclarePairedDelimiter\floor{\lfloor}{\rfloor}

\newcommand{\sign}{\mathrm{sign}}

\newcommand{\Exp}{\mathbb{E}}               
\newcommand{\E}{\mathbb{E}}                    
\newcommand{\la}{{\lambda}}                     

\newcommand{\nn}{\notag}
\newcommand{\R}{\mathbb{R}}

\newcommand{\x}{\mathbf{x}}
\newcommand{\w}{\mathbf{w}}
\newcommand{\ub}{\mathbf{u}}

\newcommand{\g}{\mathbf{g}}
\newcommand{\vb}{\mathbf{v}}

\newcommand{\e}{\mathbf{e}}

\newcommand{\z}{\mathbf{z}}

\newcommand{\ab}{\mathbf{a}}

\newcommand{\h}{\mathbf{h}}

\newcommand{\Tc}{{\mathcal{T}}}
\newcommand{\Sc}{{\mathcal{S}}}

\newcommand{\Dc}{\mathcal{D}}

\newcommand{\Kc}{\mathcal{K}}
\newcommand{\Nc}{\mathcal{N}}

\newcommand{\Lc}{\mathcal{L}}

\newcommand{\Qc}{\mathcal{Q}}
\newcommand{\Oc}{\mathcal{O}}

\newcommand{\beq}{\begin{equation}}
\newcommand{\eeq}{\end{equation}}
\newcommand{\bea}{\begin{align}}
\newcommand{\eea}{\end{align}}

\newcommand{\vp}{\vspace{4pt}}

\usepackage{amsmath}
\def\bea#1\eea{\begin{align}#1\end{align}}

\title{
The Generalized Lasso for Sub-gaussian Measurements with Dithered Quantization}
\author{
\IEEEauthorblockN{Christos Thrampoulidis and Ankit Singh Rawat}
\thanks{C.~Thrampoulidis and A.~S.~Rawat are with the Research Laboratory of Electronics, MIT, Cambridge, MA 02139, USA (e-mail: \{cthrampo, asrawat\}@mit.edu).}

{Research Laboratory of Electronics, MIT\\
Cambridge, MA 02139, USA.\\
E-mail: cthrampo@mit.edu, asrawat@mit.edu}
}
\IEEEoverridecommandlockouts

\begin{document}
\maketitle
\begin{abstract}
In the problem of structured signal recovery from high-dimensional linear observations, it is commonly assumed that full-precision measurements are available. Under this assumption, the recovery performance of the popular Generalized Lasso (G-Lasso) is by now well-established. In this paper, we extend these types of results to the practically relevant settings with quantized measurements. We study two extremes of the quantization schemes, namely, uniform and one-bit quantization; the former imposes no limit on the number of quantization bits, while the second only allows for one bit. In the presence of a uniform dithering signal and when measurement vectors are sub-gaussian, we show that the same algorithm (i.e., the G-Lasso) has favorable recovery guarantees for both uniform and one-bit quantization schemes. Our theoretical results, shed light on the appropriate choice of the range of values of the dithering signal and accurately capture the error dependence on the problem parameters. For example, our error analysis shows that the G-Lasso with one-bit uniformly dithered measurements leads to only a logarithmic rate loss compared to the full-precision measurements.
\end{abstract}

\section{Introduction}

\subsection{Motivation}

Over the last decade or so, the problem of structured signal recovery from high-dimensional linear measurements has received considerable attention. In its most classical formulation the problem asks to recover a signal $\x_0\in\R^n$ from (noisy) linear measurements $y_i=\ab_i^T\x_0 + z_i$, $i \in [m] :=\{1, 2,\ldots, m\}$, where $\ab_i,~i\in[m]$, are the measurement vectors and $z_i$ represents a noise component. In the high-dimensional setting of interest both the dimension of the signal's ambient space $n$, as well as the number of measurements $m$ are large. Moreover, there is typically some prior structural knowledge about the unknown signal $\x_0$, which manifests itself in many forms such as sparsity, low-rankness, sparse derivatives, etc.

While many recovery algorithms have been proposed and analyzed in the literature, perhaps
the most popular one is the Generalized Lasso (G-Lasso), which minimizes a least-squares objective function subject to a regularization constraint that promotes the prior structural knowledge on $\x_0$ (e.g., $\ell_1$-norm for sparse recovery and nuclear norm for low-rank matrix recovery). In its general form the G-Lasso obtains an estimate $\widehat{\x}$ by solving the following (convex) optimization program for some appropriate (convex) constraint set $\Kc$ \footnote{The discussion that follows, as well as, all the results presented in this paper do \emph{not} require the set $\Kc$ to be convex. However, convex constraint sets $\Kc$ make \eqref{eq:Lasso} a convex optimization program; thus, they are often preferred in practice for computational purposes.}.
\bea\label{eq:Lasso}
\widehat{\x} := \arg\min_{\x\in\Kc} \frac{1}{2m}\sum_{i=1}^{m} (y_i - \ab_i^T\x_0)^2.
\eea
By now, there is a very good understanding of the algorithm's recovery performance; the guarantees apply for general types of structure (and the corresponding constraint sets $\Kc$) and hold under wide range of assumptions on the measurement vectors. It is rather typical that the analysis is performed under the assumption that the measurement vectors are realized from some probability distribution. For example, suppose that the measurements are centered sub-gaussian and that the noise is sub-gaussian and independent of the measurements, with entries iid of variance $\sigma^2$. Then, with high-probability, it holds\footnote{Throughout the introduction, we are being deliberately  informal in order to streamline the presentation and focus on the main points: the symbol $``\lesssim"$ hides positive constants and ``high-probability" is not quantified. We refer the reader to the remaining sections of the paper for formal statements.}, \cite{Cha,TroppBowling}:
\bea\label{eq:LASSO_error}
\| \widehat{\x} - \x_0 \|_2 \lesssim \sigma\cdot\frac{\wg(\TcK)}{\sqrt{m}},
\eea
provided that $m\gtrsim  \wg^2(\TcK)$. In \eqref{eq:LASSO_error} the quantity $\wg(\TcK)$ is a geometric summary parameter, called the Gaussian width, which captures the geometry of the set $\Kc$ with respect to the particular $\x_0$. While deferring its formal definition to Section \ref{sec:back}, it is important for our discussion to remark the following. First, despite being an abstract parameter, it is often possible to compute sufficiently accurate approximations that reveal the explicit role of primitive problem parameters. For example, for an $s$-sparse $\x_0$ and $\Kc$ being a scaled $\ell_1$-ball, it can be shown that $\wg^2(\TcK)\leq 2s\log(n/s)+\frac{3}{2}s$. Second, the result stated above not only captures the correct error rate decay, but also, it captures the minimum required number of measurements to guarantee this decay, i.e., $m\gtrsim  \wg^2(\TcK)$.

However, in many practical settings, full-precision measurements that have been assumed thus far in this discussion are \emph{not} available. Instead, one observes \emph{quantized linear measurements} $y_i = \Qc(\ab_i^T\x_0+\z)$, where $\Qc(\cdot)$ is the quantization scheme at hand. For example, consider the following two commonly encountered quantization schemes (cf.~Figure \ref{fig:quant}).

\begin{itemize}
\item \emph{Uniform (mid-riser) quantization}: For some level $\Delta>0$,  $\Qc(x) = \Delta\big(\floor{\frac{x}{\Delta}}+\frac{1}{2}\big).$ \footnote{For $b \in \R$, $\floor{b}$ denotes the largest integer that is smaller than $b$.}

\item \emph{One-bit quantization}: $\Qc(x) = \sign(x).$
\end{itemize}

\noindent This gives rise to the following natural question:

{\emph{How to recover a structured signal $\x_0$ from high-dimensional measurements $y_i = \Qc(\ab_i^T\x_0+z_i)$? Moreover, can we obtain recovery guarantees that resemble \eqref{eq:LASSO_error}?}}

\begin{figure}[t]
\begin{minipage}[b]{.48\linewidth}
  \centering
  \centerline{\includegraphics[width = .85\linewidth]{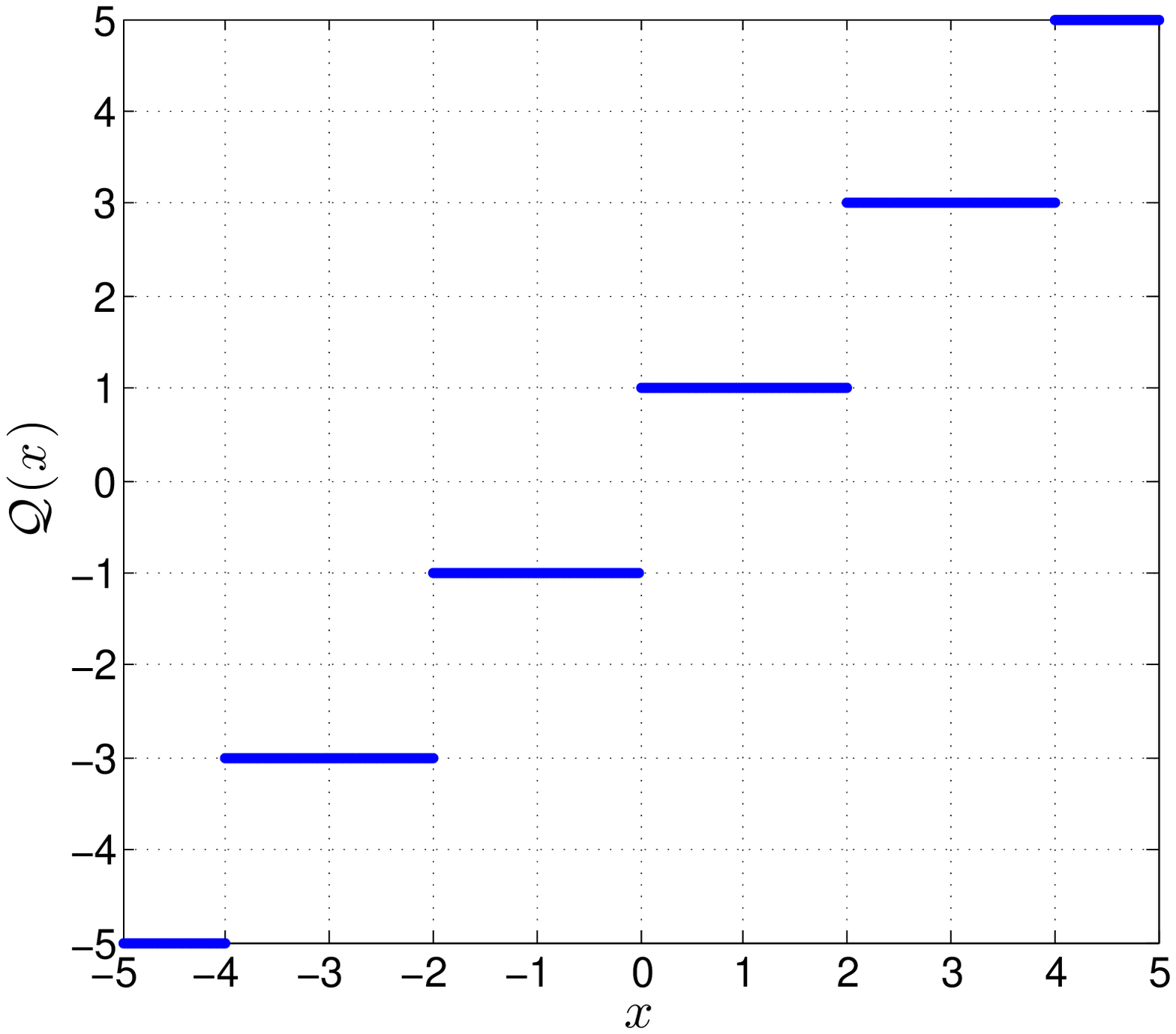}}
  \centerline{\footnotesize \qquad Uniform quantization with $\Delta = 2$.}\medskip
\end{minipage}
\begin{minipage}[b]{0.48\linewidth}
  \centering
  \centerline{\includegraphics[width = .85\linewidth]{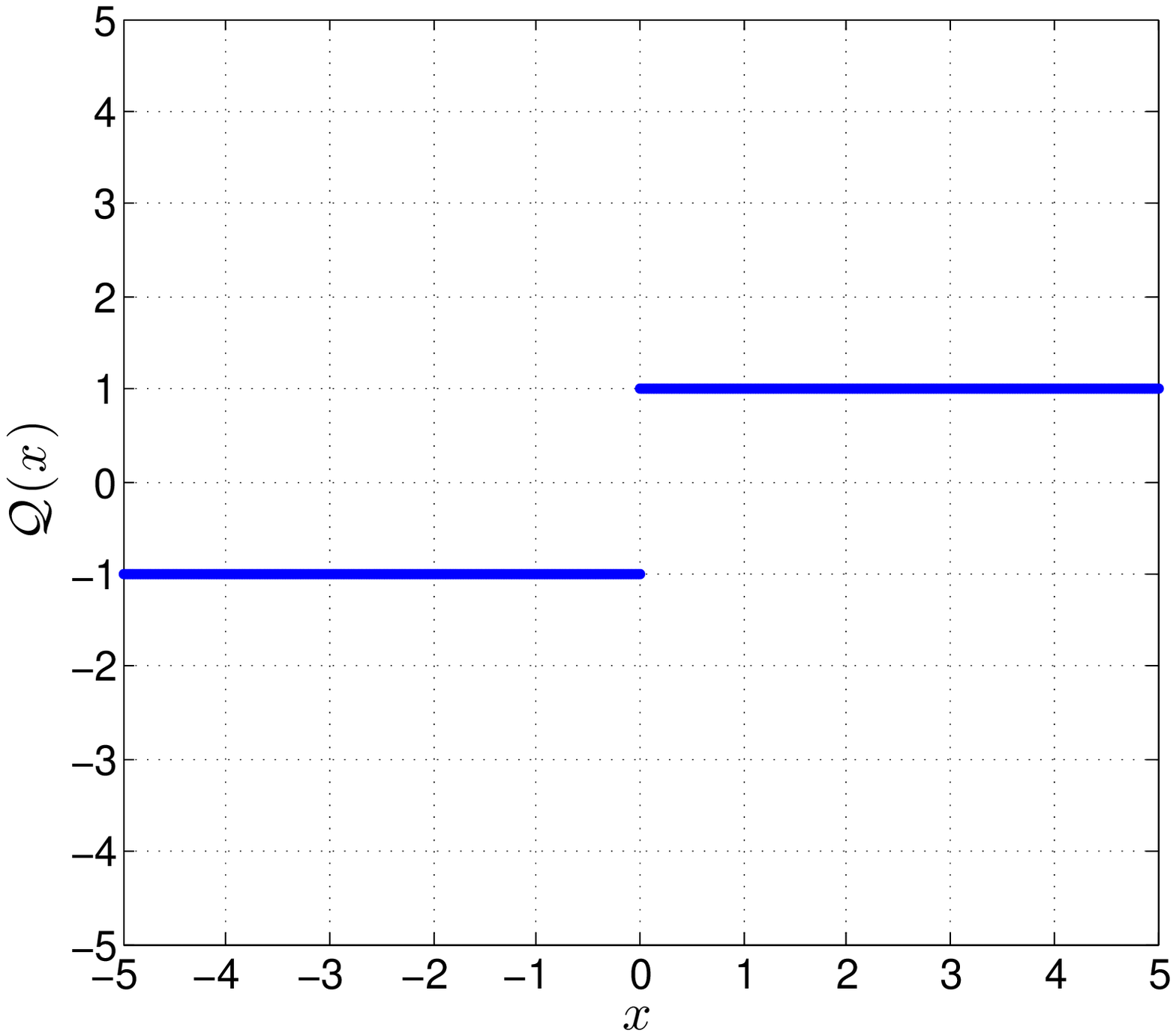}}
  \centerline{\footnotesize \qquad One-bit quantization.}\medskip
\end{minipage}
\caption{Illustration of uniform and one-bit quantization.}
\label{fig:quant}
\end{figure}

First of all, notice that the task above is naturally harder than recovery from full-precision linear measurements, as the quantized measurements are clearly less informative. A simple illustrative example is to assume that $m>n$ and the noiseless setting ($z_i=0,~i\in[m]$). Then, for the full-precision measurements, under mild assumptions, it is easy to perfectly recover $\x_0$ by just inverting the system of linear equations. On the other hand, recovery becomes challenging when only one-bit measurements $\sign(\ab_i^T\x_0)$ are available; even, in the absence of noise we may not hope to perfectly recover $\x_0$. Onwards, we focus only on the noiseless setting since it is already challenging for quantized measurements and it also leads to easier exposition.
%

Perhaps surprisingly, Plan and Versyynin \cite{plan2016generalized} demonstrated that, even with quantized measurements, the G-Lasso achieves good recovery performance when the measurement vectors are Gaussian\footnote{The question of structured signal recovery from quantized measurements (specifically, one-bit measurements) has been subject of numerous works over the past decade. We postpone a reference on this line of research  to Section \ref{sec:RW}, and instead, we focus  on the directly relevant work \cite{plan2016generalized}.}. An appealing feature of their theoretical result is that, similar to \eqref{eq:LASSO_error}, their error bounds are simple to state and clearly isolate the effect of the specific quantization scheme, on one hand, and of the problem geometry, on the other hand. To be more concrete the authors \cite{plan2016generalized} show that, with high probability, the estimator
\bea\label{eq:muLASSO}
\widehat{\x} := \arg\min_{\x\in\Kc} \frac{1}{2m}\sum_{i=1}^{m} (\mu_\Qc\cdot y_i - \ab_i^T\x_0)^2,
\eea
satisfies
\bea\label{eq:muLASSO_error}
\big\| \mu_\Qc\cdot\widehat{\x} - \frac{\x_0}{\|\x_0\|_2} \big\|_2 \lesssim \frac{\sigma_\Qc\cdot\wg(\TcK)+\eta_\Qc}{\sqrt{m}},\eea
provided that $m\gtrsim \wg^2(\TcK)$. In this expression, the non-zero parameters $\mu_\Qc, \sigma_\Qc$ and $\eta_\Qc$ depend on the specific quantization function $\Qc$. For example,  it can be shown that $\mu_\Qc=\sqrt{2/\pi}$, $\sigma_\Qc^2=1-2/\pi$ for one-bit measurements, and, $\mu_\Qc=1$, $\sigma_\Qc^2\leq \Delta^2$ for uniformly quantized measurements. For simplicity, we drop the subscript $\Qc$ when the specific scheme in reference is clear from context. 
We now make the following two crucial remarks regarding \eqref{eq:muLASSO_error}.
\begin{itemize}
\item \emph{Only-direction estimation}: It only guarantees that $\widehat{\x}$ is well-aligned with $\x_0$ and says nothing about its norm being close to that of $\x_0$ (compare the left-hand side of \eqref{eq:muLASSO_error} to the corresponding result regarding full-precision measurements in \eqref{eq:LASSO_error}).

\item \emph{Gaussian measurements}: The validity of \eqref{eq:muLASSO_error} requires that the measurement vectors are Gausssian. 
\end{itemize}

\noindent In this work, we address both of the aforementioned weaknesses. We show that in the presence of appropriate \emph{dithering} in the quantization scheme, the same algorithm \eqref{eq:muLASSO} can recover not only direction but also norm information for distributions of the measurement vectors beyond Gaussians. Our results apply to both of the quantization schemes mentioned above, and they also inherit many of the  interpretability features of \eqref{eq:LASSO_error} and \eqref{eq:muLASSO_error}.

\subsection{Contribution}

We consider quantized linear measurements with appropriate \emph{dithering}. Dither is a purposeful applied random noise component that is added to an input signal prior to its quantization. This technique is rather well-established and commonly used both in practice (because it can result in more subjectively pleasing reconstructions) and in theory (because it often results in favorable statistical properties of the quantization noise); e.g., see \cite{gray1993dithered,dabeer2006signal} and references therein. More recently, dithered quantization has been also exploited and studied in the context of high-dimensional structured signal recovery from quantized linear measurements \cite{baraniuk2017exponential,knudson2016one,xu2018taking,dirksen2018robust}. Our work builds upon such recent results, in particular \cite{xu2018taking,dirksen2018robust}; see Section \ref{sec:RW} for a detailed discussion.

In a nutshell, we show that the G-Lasso  can be used to efficiently recover structured signals from (appropriately) dithered quantized linear measurements. More precisely, we study the recovery method in \eqref{eq:muLASSO} for appropriate value of the parameter $\mu_Q$ when the measurements are of the form $y_i=\Qc(\ab_i^T\x_0+\tau_i)$, where $\tau_i$ is the dithering signal and $\Qc$ is either the uniform or the one-bit quantizer. We consider sub-gaussian measurement vectors of sub-gaussian norm at most $L$ (see~Section \ref{sec:back}).

Out results are rather easy to state. We include here an informal version to allow direct comparisons to \eqref{eq:LASSO_error} and \eqref{eq:muLASSO_error}. Formal statements and detailed discussions follow in later sections.

\paragraph{Uniform dithered quantization} For $i \in [m]$, let  the measurements $y_i$ be given as follows:
\bea\label{eq:uni_model}
y_i =  \Delta\Big(\Big\lfloor{\frac{\ab_i^T\x_0 + \tau_i}{\Delta}}\Big\rfloor+\frac{1}{2}\Big),
\eea
where $\tau_i\sim\mathrm{Unif}\big(-\frac{\Delta}{2},\frac{\Delta}{2}\big]$ and $\Delta>0$. We solve the G-Lasso in \eqref{eq:muLASSO} setting $\mu_\Qc=1$ and we show that,
\bea\label{eq:error_uni}
\| \widehat{\x} - \x_0 \|_2 \lesssim \Delta \cdot \frac{\wg(\TcK)}{\sqrt{m}}, 
\eea
provided that $m\gtrsim \wg^2(\TcK)$. Note the resemblance of this result to \eqref{eq:LASSO_error}. We may conclude that essentially the G-Lasso treats the quantization error (up to an absolute constant) as an independent noise component of strength $\Delta$ \footnote{It is rather straightforward to see that the quantization error is a random variable that is absolutely bounded by $\Delta$ under the quantization scheme in \eqref{eq:error_uni}. On the other hand, this random variable is \emph{not} independent of the measurements. Thus, \eqref{eq:LASSO_error} is not applicable and showing that \eqref{eq:error_uni} holds requires additional effort}. Also, observe that unlike \eqref{eq:muLASSO_error} our guarantee does not require knowledge of the norm of the signal $\|\x_0\|_2$.

\paragraph{One-bit dithered quantization.} For $i \in [m]$, let  the measurements $y_i$ be given as follows:
\bea\label{eq:one_model}
y_i = \sign(\ab_i^T\x_0 + \tau_i),
\eea
where $\tau_i\sim\mathrm{Unif}[-T,T]$  and  $T>0$.
Furthermore, assume a known upper bound on the true norm of signal $\x_0$, i.e., known $R>0$ such that $\|\x_0\|_2\leq R$.
 Our main result suggests setting $T=c \,L\, R \, \sqrt{\log{m}}$ for some absolute constant $c$ and solving  \eqref{eq:muLASSO} while setting $\mu_\Qc=T$. Then, we show that
\bea\label{eq:error_one}
\| \widehat{\x} - \x_0 \|_2 \lesssim LR\,\Big(\sqrt{\log m} \,\maxx \sqrt{\log n}\,\Big)\cdot\frac{\wg(\TcK) + C_2}{\sqrt{m}}, 
\eea
provided that $m\gtrsim \wg^2(\TcK)$.  Here, $a\maxx b := \max\{a,b\}.$
Observe that despite only having one-bit information, the proposed estimator leads to only a logarithmic loss with respect to the best possible error rate  in \eqref{eq:LASSO_error}. 

In the coming sections we also discuss (possible) extensions of this theory to other dithering distributions and to other random measurement models. We also corroborate our theoretical findings with numerical simulations.

%



\subsection{Proof sketch}
\label{sec:sketch}

\paragraph{A simple key inequality} We start by denoting the loss function in \eqref{eq:muLASSO} as
$$
\Lc(\x) := \frac{1}{2m}\sum_{i=1}^{m}{(\mu_\Qc y_i-\ab_i^T\x_0)^2}.
$$
Let $\widehat{\x}$ be a solution of \eqref{eq:muLASSO} and $\widehat{\w} = \widehat{\x} - \x_0$ denote the error vector. If $\hat\w=0$ there is nothing to prove and so we assume onwards that this is $\emph{not}$ the case. By optimality of $\widehat{\x}$ (consequently, of ${\x_0} + \widehat{\w}$), we have that
\bea\nn
0 &\geq \Lc(\x_0+\wh) - \Lc(\x_0) \nonumber\\
 & = \frac{1}{2m}\sum_{i = 1}^{m} \big(\ab_i^T\widehat{\w}\big)^2 - \frac{1}{m} \sum_{i = 1}^{m}\big(\mu_\Qc y_i - \ab_i^T\x_0\big)\ab_i^T\widehat{\w},
\eea
where the equality holds by simple algebraic manipulations. Rearranging that expression yields
\bea \nn
\frac{1}{m}\sum_{i = 1}^{m} \big(\ab_i^T\widehat{\w}\big)^2 \leq \frac{2}{m} \sum_{i = 1}^{m}\big(\mu_\Qc y_i - \ab_i^T\x_0\big)\ab_i^T\widehat{\w},  
\eea
which is our starting point to obtain an upper bound on $\|\wh\|_2$. Recall that $\widehat{\w}\in\Kc-\x_0$, which leads to ${\widehat{\w}}/{\|\widehat{\w}\|_2}\in\mathrm{cone}(\Kc-\x_0)=:\Dc(\Kc,\x_0)$, where $\Dc(\Kc,\x_0)$ denotes the cone of descent directions or {\em tangent cone} (cf.~Definition \ref{def:tangent}). Therefore, we can deduce the following key fact:
\bea \label{eq:key}
&\|\wh\|_2\cdot \underbrace{\inf_{\w\in \Sc^{n-1}\cap\Dc(\Kc,\x_0)}\frac{1}{m}\sum_{i = 1}^{m} \big(\ab_i^T{\w}\big)^2}_{:=\mathrm{LB}} \leq \nonumber \\
&~~~~~~2\cdot \underbrace{\sup_{\w\in \Sc^{n-1}\cap\Dc(\Kc,\x_0)}\frac{1}{m} \sum_{i = 1}^{m}\big(\mu_\Qc y_i - \ab_i^T\x_0\big)\ab_i^T\w}_{:=\mathrm{UB}},  
\eea
where $\Sc^{n-1}$ denotes the unit sphere in $\R^n$. The rest of the proof amounts to lower (upper) bounding the $\inf$ $(\sup)$ of the involved random processes in \eqref{eq:key}. Applying those bounds in \eqref{eq:key} naturally leads to the results stated in \eqref{eq:error_uni} and \eqref{eq:error_one}.


\paragraph{Lower bound.} The desired lower bound on $\mathrm{LB}$ follows directly by applying Mendelson's small ball method \cite{mendelson1,mendelson2,TroppBowling}. While deferring the details to later sections, we mention that requiring that $\mathrm{LB} >c>0$ results in the ``provided that $m\gtrsim \wg^2(\TcK)$" part of the results.

\paragraph{Upper bound} We will show how to upper bound $\Exp[\mathrm{UB}]$ in \eqref{eq:key}, where the expectation is over all the involved random variables, i.e., the $\ab_i$'s and $\tau_i$'s. This leads to a constant probability bound (say with probability at least $0.995$) by Markov inequality, but more powerful techniques can also be applied to yield similar bounds that hold up to probability of failure that goes to zero with increasing number of measurements. For the ease of exposition, we denote the \emph{quantization noise} as follows:
\bea\label{eq:q_error}
e_i := \mu_{\Qc} y_i - \ab_i^T\x_0 = \mu_{\Qc}\cdot\Qc(\ab_i^T\x_0+\tau_i) - \ab_i^T\x_0.
\eea 
In the sequel, keep in mind that $e_i$ is \emph{not} independent of the measurement vectors $\ab_i$.
Also, for a random process $X_{\w}$ indexed by $\w$, let us denote 
\bea\label{eq:norm_def}
\| X_\w \| := \sup_{\w\in \Sc^{n-1}\cap\Dc(\Kc,\x_0)} X_\w.
\eea

\noindent By observing that $\| X_\w + X'_\w \|\leq \| X_\w\| +  \| X'_\w\|$, it easily follows that
\bea
&\Exp[\mathrm{UB}] \leq \Exp\big\| \frac{1}{m} \sum_{i=1}^{m} (e_i\mathbf{\ab_i}^T\x_0  - \Exp[ \tilde{e}_i{\tilde{\ab}_i}^T\w ] )\big\| + \nonumber \\
&~~~~~~~~~~~~\big\|\frac{1}{m}\sum_{i=1}^{m}\Exp[ \tilde{e}_i\mathbf{\tilde{\ab}_i}^T\x_0 ] \big\| \nn
\\ 
&=  \underbrace{\Exp\big\| \frac{1}{m} \sum_{i=1}^{m} (e_i\mathbf{\ab_i}^T\x_0  - \Exp[ \tilde{e}_i{\tilde{\ab}_i}^T\w ] )\big\|}_{\mathrm{Term~I}} + \underbrace{\big\| \Exp[ {e}_1{\ab}_1^T\w ] \big \|}_{\mathrm{Term~II}},\label{eq:Terms}
\eea
where, for each $i \in [m]$,  $\tilde{e}_i$ and $\tilde{\ab}_i$ are iid copies of $e_1$ and $\ab_1$, respectively. Now, we need to show that both {\rm Term~I} and {\rm Term II} are small. We may think of these as a \emph{bias} ({\rm Term II}) and \emph{variance} ({\rm Term I}) terms. Appropriately selecting the dithering signal helps reduce the bias in the estimate. At the same time, since the dither signal is not known, it acts as a source of noise and naturally increases the variance of the estimate.

Here, in oder to keep this proof sketch short, we focus on Term II. This alone already demonstrates the value of dithering and guides the correct choice of the parameter $\mu_\Qc$.
%
%
 The details regarding bounding Term I can be found in later sections.
  We consider Term II separately for each one of the quantization schemes that we wish to analyze.

\noindent{\emph{Uniform quantization.}} The key observation here is that for all inputs $x\in\R$ the quantization noise of a uniform quantizer with uniform dithering (see \eqref{eq:uni_model}) is a mean zero random variable, i.e.,
\bea\label{eq:mean_uni}
\E_{\tau} [Q(x+\tau) - x] = 0.
\eea
This is a classical fact in the theory of dithered quantization (see Section \ref{sec:uni} for a discussion) \cite{gray1993dithered}. By using this fact and the tower property of expectation, one easily finds that {\rm Term II} is equal to zero:
\bea\label{eq:uni_zero}
\E[e_1(\ab_1^T\w)] =\E_{\ab_1}\big[\E_\tau[e_1](\ab_1^T\w)\big] = 0.
\eea

\noindent{\emph{One-bit quantization.}} As compared to the uniform quantization, the quantization noise in the case of $1$-bit quantization is not zero-mean. 
However, with an appropriate choice of $T$ we can still make {\rm Term II} small enough. A simple calculation yields the following for the quantization scheme in \eqref{eq:one_model}:
\bea\nn
\E_{\tau} [Q(x+\tau)] = \frac{x}{T} - \frac{x}{T}\ind{|x|>T} + \ind{x>T} - \ind{x<-T}.
\eea
Hence, by choosing $\mu_\Qc=T$, we have that 
\bea\label{eq:mean_one}
\E_{\tau} [\mu_\Qc Q(x+\tau) - x] &= -x\ind{|x|>T} + T\ind{x>T} -  \nonumber \\
&~~~T\ind{x<-T}.
\eea
In \eqref{eq:Terms} the role of $x$ above is played by $\ab_1^T\x_0$. Clearly the right-hand side in \eqref{eq:mean_one} is non-zero for general values of $x$, but we can hope of making it small by choosing $T$ large enough so that the events under which the indicator functions become active are rare. To see this, recall our assumption that $\ab_1$ is isotropic $L$-subgaussian, from which it follows that $\Pr(|\ab_i^T\x_0|>t)\leq 2\exp(-ct^2/(L^2\|\x_0\|^2)).$ Notice that this probability can be made sufficiently small by setting $T=c L R \sqrt{\log{m}}$. Of course, a little more work is needed to translate this into $\E[e_1(\ab_1^T\w)]$ being small, but at this point we defer the rest of the details to later sections (see Lemma \ref{lem:upper_two}).
%

\begin{remark}[Literature]
The method for analyzing the G-Lasso performance based on \eqref{eq:key} has been commonly used in several recent works. In fact, this is the starting point not only for the analysis under dithered quantized measurement, but also for the error bounds in \eqref{eq:LASSO_error} and \eqref{eq:muLASSO_error}, e.g. \cite{plan2016generalized}. Beyond that, as previously mentioned, the lower bound is based on Mendelson's small-ball method \cite{mendelson1,mendelson2}. For the upper bound: (Term I) we carefully put together several known techniques in the study of suprema of random processes (such as symmetrization, Rademacher contraction principle, majorizing measure theorem, etc.; see Lemmas \ref{lem:upper_uni} and \ref{lem:upper_one});  (Term II) we exploit the fact that dithering causes the quantization noise to behave in a statistically nice fashion. Although this latter idea is well-known, in the context of our paper, it was brought to our attention by the  recent works \cite{dirksen2018robust,xu2018taking}. More precisely: (i) Identity \eqref{eq:mean_uni} is the key fact used in \cite{xu2018taking} (but also, see earlier classical works on dithered quantization, e.g., \cite{gray1993dithered}); (ii) Identity \eqref{eq:mean_one} is previously derived and exploited in the same way in \cite{dirksen2018robust} (but also, see earlier work \cite{dabeer2006signal})
\end{remark}

%


\subsection{Related Work}
\label{sec:RW}

Our work naturally fits in the recent developments in the study of structured signal recovery from high-dimensional random measurements. With the advent of Compressive Sensing (CS), there has been a very long list of papers that collectively have significantly advanced our understanding regarding the performance of convex-optimization based methods in the case of (full-precision) noisy linear measurements. Perhaps the most widely used and most well-studied among such methods is the Generalized Lasso in \eqref{eq:Lasso} (and its variants). By now, there exists a rich, elegant and general theory that accurately (only up to absolute constants) characterizes the recovery performance of the G-Lasso under quite general assumptions on the measurements vectors (iid Gaussians, sub-gaussians, sub-exponentials, etc.), e.g., see \cite{rudelson2006sparse,Sto,donoho2011noise,Cha,TroppEdge,StoLASSO,OTH13,TroppBowling,sivakumar2015beyond,oymak2015universality}. In this paper, we extend this line of work by establishing recovery guarantees for the G-Lasso in the practical settings with quantized measurements. The error bounds that we derive are reminiscent of the existing results in the case of the full-precision measurements.

Structured signal recovery from quantized high-dimensional measurements has also been extensively studied in the literature. The vast majority of the related works focuses on the case of one-bit quantization (often termed 1-bit CS) \emph{without} dithering, for which case norm recovery is impossible, e.g., see \cite{boufounos20081,plan2013one,jacques2013robust}. Also, most of these works, only apply to iid Gaussian measurements, with a few exceptions such as \cite{ai2014one}. On the other hand, it was recently demonstrated that dithering has the advantage of making norm-recovery possible: \cite{knudson2016one} considers iid Gaussian dithering signal, while \cite{baraniuk2017exponential} studies an adaptive dithering scheme. Both of these works are limited to iid Gaussian measurements and sparse signal recovery. It has only been very recent work due to  Xu and Jacques \cite{xu2018taking}, who (to the best of our knowledge) first demonstrated that a uniform quantization scheme combined with a uniformly distributed dithering signal promises pushing much of the theory beyond Gaussian measurements. Shortly afterwards, Dirksen and Mendelson \cite{dirksen2018robust} extended this idea to one-bit measurements with appropriate uniform dithering. These two papers have motivated our work. We show that a single recovery algorithm can successfully be used for both quantization schemes considered in \cite{xu2018taking} and \cite{dirksen2018robust}. Importantly, this algorithm is the well-established G-Lasso algorithm with a single tuning parameter $\mu_\Qc$, which changes depending on the specifics of the quantization scheme. In terms of theory, our analysis yields easily interpretable results that nicely fit in the existing literature on the full-precision measurements. Practically, the potential advantage of using the G-Lasso is that one can rely on the abundance of efficient specialized solvers for this program. We empirically observe that the G-Lasso significantly outperforms the simple recovery scheme proposed in \cite{xu2018taking} for uniform quantization. On the other hand, the G-Lasso appears to perform similarly to the algorithm proposed in \cite{dirksen2018robust}. Yet, the former has the advantage of being directly applicable to uniformly quantized measurements. Also, our analysis suggests explicit guidelines on the choice of the threshold $T$, which controls the range of dithering \eqref{eq:one_model}, and of the tuning parameter $\mu_\Qc$ in \eqref{eq:muLASSO}. 

Finally, our paper is very closely related to the work of Plan and Vershynin \cite{plan2016generalized}, who studied the G-Lasso for non-linear observations, which includes quantization as a special case. Unfortunately, their results do not guarantee norm-recovery and are limited to Gaussian measurements. Our work removes these limitations in the case of quantized measurements. 

\subsection{Organization}
The rest of the paper is organized as follows. In Section~\ref{sec:back}, we introduce various key geometric quantities that are relevant to our analysis and state the underlying assumptions on the measurement vectors. We present our main results and accompanying discussions for the uniform dithered quantization and the one-bit dithered quantization models in Section~\ref{sec:uni} and Section~\ref{sec:one}, respectively. In Section~\ref{sec:simulations}, we evaluate the recovery performance of \eqref{eq:muLASSO} using synthetic experiments and compare it with the existing methods in the literature. We conclude the paper by highlighting multiple concrete directions for future work in Section~\ref{sec:conclude}. 
We have relegated all the proof to appendices to enhance the readability of the paper.

\section{Background}\label{sec:back}

\subsection{Geometric notions}
%

First, we introduce the notions of the tangent cone and the Gaussian width.
\begin{defn}[Tangent cone]\label{def:tangent} The \emph{tangent cone} of a set $\Kc\subset\R^n$ at $\x\in\R^n$ is defined as
$$\Dc(\Kc,\x):=\{\la \vb: \la\geq0, \vb\in\Kc-\x\}.$$
\end{defn}

\begin{defn}[Gaussian width]
The {\em Gaussian width} $\omega(\Tc)$ of a set $\Tc \subset \R^n$ is defined as
\begin{equation}\label{eqn:Gauss_width}
			\omega(\Tc) := \E \big[	\sup_{\vb \in \Tc } \g^T\vb \big],	\quad \g \sim\Nc(0,\Id_n).
		\end{equation}
\end{defn}

The Gaussian width 
plays a central role in asymptotic convex geometry. In particular, its square $\wg^2(\Tc)$ can be formally described as a measure of the effective dimension of the set $\Tc$ \cite{VerBook,TroppEdge}. More recently, the Gaussian width has played a key role in the study of linear inverse problems. This is already revealed in \eqref{eq:LASSO_error} which requires that the number of measurements $m$ be larger than (a constant multiple of) the squared Gaussian width of a spherical section of the corresponding descent cone $\TcK:=\Dc(\Kc,\x_0)\cap\Sc^{n-1}$ \cite{Sto,Cha,TroppBowling}.

Importantly, this line of work has resulted in the development of principled recipes that yield useful numerically satisfactory bounds on the Gaussian width \cite{Sto,Cha,TroppEdge,OTH13}.

\subsection{Sub-gaussian vectors}
\label{sec:probability}

Throughout this paper, we work with sub-gaussian measurement vectors. For the reader's convenience, we recall the definition of sub-gaussian vector below; see e.g., \cite[Ch.~2]{VerBook} for an introduction to sub-gaussian random variables.

\begin{defn}[Sub-gaussian vectors]\label{def:sub}
A random vector $\h\in\R^n$  is called sub-gaussian with sub-gaussian norm $\psin{\h}$ if the one-dimensional marginals $\h^T\x$ are sub-gaussian random variables for all $\x\in\Sc^{n-1}$ and $$\psin{\h} = \sup_{\x\in\Sc^{n-1}}\psin{\h^T\x},$$
where we recall that the sub-gaussian norm $\psin{X}$ of a random variable $X$ is defined as follows
\bea\nn
\psin{X}:=\inf\{t>0 : \E \exp(X^2/t^2)\leq 2\}.
\eea
\end{defn}

Specifically, we make the following assumption on the measurement vectors $\ab_i,~i\in[m]$.

\begin{assumption}[Sub-gaussian measurements]\label{ass:sub}
 We assume that each vector $\ab_i,~i\in[m],$ is an iid copy of a random vector $\ab\in\R^n$ that satisfies the following properties.
\begin{itemize}
\item {\em Subgaussian marginals}: $\ab$ is a sub-gaussian random vector with $\psin{\ab}=L$.
\item {\em Symmetry}: $\ab$ has a symmetric distribution. In particular this implies $\E\ab = 0$.
\item {\em Nondegeneracy}: There exists $\alpha > 0$ such that for each $\ub \in \Sc^{n-1}$, we have $\mathbb{E}|\ab^T\ub| \geq \alpha$.
\end{itemize}
\end{assumption}

\section{Uniform quantization}\label{sec:uni}
In this section we study the problem of exact signal recovery from the observations generated by the uniform dithered quantization. In particular, we aim to recover a signal of interest $\x_0$ from $m$ measurements $y_i$ given by \eqref{eq:uni_model}.  Our estimate of $\x_0$ solves \eqref{eq:muLASSO} with parameter $\mu_\Qc = 1$. We assume that the measurement vectors satisfy Assumption \ref{ass:sub}. 
Our main result is as follows; we defer its proof to Appendix \ref{sec:proofs_uni}.

\begin{thm}\label{thm:uni_main}[Error analysis:~uniform dithered quantization]
Suppose that the vectors $\ab_i\in\R^n, i\in[m],$ satisfy Assumption \ref{ass:sub}. For a fixed vector $\x_0 \in \R^n$, let the measurements $y_i, i\in[m],$ be given as in \eqref{eq:uni_model} and let $\widehat{\x}$ be a solution to \eqref{eq:muLASSO} with parameter $\mu_{\Qc}=1$. Finally, for the constraint set $\Kc$ in \eqref{eq:muLASSO}, define the shorthand $\TcK:=\Sc^{n-1}\cap\Dc(\Kc,\x_0)$. 
Then, there exist positive constants $C,c_1=c_1(L/\alpha), c_2 = c_2(L/\alpha)$ such that the following holds with probability at least 0.99:
\bea\nn
\|\widehat{\x} - \x_0\|_2\leq C\Delta \cdot\frac{\wg(\TcK)}{\sqrt{m}},
\eea
provided that
\bea\label{eq:mgeq}
m \ge c_1 \cdot \wg^2(\TcK) + c_2.
\eea
\end{thm}

\subsection{Remarks and Extensions}
\label{rem:uni}

\paragraph{On the constant probability bound} As stated, the error bound of the theorem holds with constant success probability. However, it is possible to extend the result to hold with  success probability that goes to one as the dimension of the problem increases. This requires applying a few technical results on concentration properties of the suprema of random processes (e.g., \cite{dirksen2015tail}). Since this does not contribute to the essence of our results, we have decided to keep the exposition simple by focusing only on the constant probability bounds (similar to \cite{plan2016generalized}).

\paragraph{Limit of full-resolution measurements} In the limit of the resolution of the quantizer $\Delta\rightarrow0$, the quantized measurements in \eqref{eq:uni_model} approach the (noiseless) full-resolution measurements $y_i=\ab_i^T\x_0,~i\in[m].$ In that limit, Theorem \ref{thm:uni_main} suggests that $\|\x-\x_0\|_2\rightarrow0$ provided that $m \geq \Oc( \omega^2(\TcK) )$. As expected, this conclusion is in full-agreement with the well-established results on the phase-transition of noiseless linear inverse problems with full-resolution measurements, e.g., \cite{TroppBowling}.

\paragraph{On the distribution of the dithering signal} As previously mentioned, dithering is essential for the validity of the theorem. This is also revealed in the proof, where the specific uniformly distributed dithering signal guarantees that the quantization noise has zero mean conditioned on the input signal (cf.~\eqref{eq:uni_zero}). This raises a natural question: what are other dithering distributions (other than uniform) that guarantee that \eqref{eq:uni_zero} holds. This question is well-studied in the literature of dithered quantization; in fact, the entire class of such distributions is characterized and we refer the interested reader to the excellent exposition in \cite[Thm.~2]{gray1993dithered} for details. As an example, \eqref{eq:uni_zero} also holds when $\tau_i,~i\in[m]$ are iid and distributed as the sum of $k\geq 1$ uniform random variables in $(-\Delta/2,\Delta/2]$. To further relate this classical result to Theorem \ref{thm:uni_main}, note that the theorem essentially remains valid for all such distributions for the dither signal with bounded support. 

\paragraph{Examples}  The results of Theorem \ref{thm:uni_main} apply under rather general assumptions on $\x_0$ and on the choice of the constraint set $\Kc$. Here, for mere illustration, we provide two concrete examples for the most popular instances of structured signal recovery problems.

\begin{itemize}
\item \emph{Spare recovery}: Assume that $\x_0\in\R^n$ is $s$-sparse, i.e., $\|\x_0\|_0=s$, and that we further choose $$\Kc=\{\x : \|\x\|_1\leq \|\x_0\|_1\}.$$ In this case, it is well-known (e.g., \cite{Cha}) that 
\bea\label{eq:gw_sparse}
\wg^2(\TcK)\leq 2s\log\Big(\frac{n}{s}\Big) + \frac{3}{2}s.
\eea
Directly applying this result to Theorem \ref{thm:uni_main} proves that
$$
\|\widehat{\x}-\x_0\|_2 \leq \Delta\cdot\Oc\Big(\frac{\sqrt{s\, \log({n}/{s})}}{\sqrt{m}} \Big),
$$
provided that $m\geq \Oc\big(s\, \log\big(\frac{n}{s}\big) \big).$

\item \emph{Low-rand recovery}: Assume that $\x_0=\operatorname{vec}(X_0)\in\R^{n^2}$, where $X_0\in\R^{n\times n}$ has rank $r$. We further choose $$\Kc=\{X\,|\,\|X\|_*\leq \|X_0\|_*\},$$ where $\|\cdot\|_*$ denotes the nuclear norm. In this case, it is well-known (e.g., \cite{Cha}) that 
\bea\label{eq:gw_lowrank}
\wg^2(\TcK)\leq 6nr.
\eea
Thus,
$$
\|\widehat{\x}-\x_0\|_2 \leq C\,\Delta \,  \frac{\sqrt{n r}}{\sqrt{m}},
$$
provided that $m\geq \Oc\big(n r\big).$
\end{itemize}

\paragraph{Boundary of $\Kc$} In the examples provided above, the set $\Kc$ is chosen such that $\x_0$ lies on its boundary. This condition is, in general, a prerequisite so that the cone of descent directions $\Dc(\Kc,\x_0)$ is \emph{not} the entire space and that the upper bound of Theorem \ref{thm:uni_main} in terms of $\wg(\TcK)$ is especially useful. Otherwise, $\TcK=\Sc^{n-1}$ and  $\wg(\TcK)\approx\sqrt{n}$, in which case the error bound fails to capture the role of $\Kc$ and of the structure of $\x_0$. When $\Kc-\x_0$ is a star-shaped set (in particular, this holds when $\Kc$ is convex), then it is possible to break that limitation of Theorem \ref{thm:uni_main} by only slightly modifying the proof and by introducing the ``local Gaussian width" in place of the Gaussian width considered here. The technical arguments towards these modifications are well-explained in \cite[Thm.~1.9]{plan2016generalized}. The focus of Theorem \ref{thm:uni_main} (also, of Theorem \ref{thm:one_main}) is on capturing the role of dithered quantized measurements on the recovery performance of the G-Lasso. Hence, we refer the reader to the related works \cite{plan2016generalized} for extensions regarding capturing the role of $\Kc$ and of $\x_0$.


\paragraph{Sub-exponential measurements} It is possible to extend the result of Theorem \ref{thm:uni_main} (with appropriate modifications on the error bound) to a wider class of measurement vectors, in particular to sub-exponential distributions. Indeed, a close inspection of the proof of Theorem \ref{thm:uni_main} reveals that the fact that measurement vectors follow a sub-gaussian distribution is essentially only critical in lower bounding the left-hand side of \eqref{eq:key} and in upper bounding the empirical width $\E\sup_{\w\in\TcK}\sum_{i\in[m]}\frac{1}{\sqrt{m}}\varepsilon_i\ab_i^T\w$ (see Appendix \ref{sec:proofs_uni_upper}). Since corresponding lower and upper bounds are also available for sub-exponential measurement vectors \cite{TroppBowling,oymak2018learning}, it is possible to extend Theorem \ref{thm:uni_main} in that direction.

\paragraph{Related results} 
Essentially, Theorem \ref{thm:uni_main} can be viewed as an extension of the corresponding results in \cite{plan2016generalized} (which are only true for Gaussians measurements) to sub-gaussian measurements\footnote{Of course, Plan and Vershynin \cite{plan2016generalized} study the generalized linear measurement model to which the quantized measurement model is only a special case. Note that the authors state their result under the additional assumption that $\|\x_0\|_2=1$ (see \eqref{eq:muLASSO_error}). However, in the case of uniform dithered quantization with uniformly distributed dither signal, it is relatively easy to extend their result by waiving the $\|\x_0\|_2=1$ assumption. We omit the details for brevity.}. Xu and Jacques \cite{xu2018taking} were the first to study the effect of dithering in uniform quantization schemes in the high-dimensional setting. To showcase the favorable properties of dithering they analyzed the performance of a simple	projected back projection (PBP) method. Naturally, as also confirmed via simulations in Section \ref{sec:simulations}, our proposed algorithm in \eqref{eq:muLASSO} outperforms PBP method. On the other hand, using a different type of analysis, Xu and Jacques are able to extend their result to wider classes of measurement distributions beyond sub-gaussians.


\vp

\paragraph{On the number of quantization bits} The uniform quantization scheme studied thus far does \emph{not} assume any constraints on the number of bits used for quantization of the input signal $\{\ab_i^T\x_0\}_{i\in[m]}$. In general, the input signal can take very large values (relative to the resolution value $\Delta$) and thus it may require a large number of quantization bits, which might be impractical. A possible solution to this issue is to perform clipping, that is to limit the number of quantization levels to some fixed value $N$. Unfortunately, clipping introduces an additional source of error to the measurement model (often referred to as overload distortion). Hence, it is not obvious at the outset how this affects the error bound of Theorem \ref{thm:uni_main}. In the next section, we study one-bit quantization, which can be viewed as an extreme case of clipping ($N=1$). Extending these results to the quantizers with general values of $N$ is an interesting direction for future research. 


\section{One-bit quanitzation}\label{sec:one}

%

We explore the problem of exact signal recovery from one-bit dithered observations when the measurement vectors are sub-gaussian. Recall that we are working with the $m$ measurements $y_i$ of a signal of interest $\x_0$ given by \eqref{eq:one_model}. The dithered signal is uniform in $[-T,T]$, and our main result specifies appropriate values for the range parameter $T>0$. Our estimate of $\x_0$ solves \eqref{eq:muLASSO} with parameter $\mu_\Qc = T$. 

We present the main result of this section in Theorem \ref{thm:one_main} below. All the proofs are deferred to Appendix \ref{sec:proofs_one}.

\begin{thm}\label{thm:one_main}[Error analysis: one-bit dithered quantization]
Suppose that the vectors $\ab_i\in\R^n, i\in[m]$, satisfy Assumption \ref{ass:sub}. Fix any $\x_0\in\R^n$ and let $R>0$ be such that $\|\x_0\|_2\leq R$. Assume that the measurements $y_i, i\in[m]$ are given as in \eqref{eq:one_model} with range parameter $T>0$ and let  $\widehat{\x}$ be a solution to \eqref{eq:muLASSO} with parameter $\mu_{\Qc}=T$. Finally, for the constraint set $\Kc$ in \eqref{eq:muLASSO} define the shorthand  $\TcK:=\Sc^{n-1}\cap\Dc(\Kc,\x_0)$. 
Then, there exist positive constants $C_1,C_2,C_3, c_1 := c_1(L/\alpha), c_2 := c_2(L/\alpha)$ such that the following holds with probability at least 0.99:
\bea\nn
\|\widehat{\x} - \x_0\|_2& \leq C_1LR\,\Big(\sqrt{\log m} \,\maxx \sqrt{\log n}\,\Big)\cdot\frac{\wg(\TcK) + C_2}{\sqrt{m}},
\eea
provided that ~$T = C_3LR\sqrt{\log{m}}$~ and~
$m \ge c_1 \cdot \wg^2(\TcK) + c_2.$
\end{thm}

\subsection{Remarks}
\label{rem:1bit}

\paragraph{On the error decay with $m$}Theorem \ref{thm:one_main} guarantees an error decay $\Oc(\sqrt{\frac{\log{m}}{m}}\,)$ as a function of the number of measurements. Therefore, an important conclusion of the theorem is that the proposed LASSO estimator in \eqref{eq:muLASSO} only leads to $\sqrt{\log{m}}$ loss with respect to the best possible rate $1/\sqrt{m}$ (see \eqref{eq:LASSO_error}). In other words, despite only having (appropriately dithered) one-bit measurements there is relatively little loss in performance with respect to full precision samples.
%

\paragraph{Gaussian measurements} When the measurement vectors have entries iid Gaussian, it is possible to perform a tighter analysis, which is presented in Appendix \ref{sec:gaussian} and which leads to the following result. 
\begin{thm}[Gaussian measurements]\label{thm:one_gauss}
Let the same setting as in Theorem \ref{thm:one_main} except that the measurement vectors $\ab_i,~i\in[m]$, are assumed to have iid standard normal entries. Then, for sufficiently large $m$, there exist absolute constants $C_1,C_2,c>0$ such that the following holds with probability at least 0.99: 
\bea\nn
\|\widehat{\x} - \x_0\|_2\leq C_1\, R\sqrt{\frac{\log{m}}{m}}\cdot\big(\wg(\TcK)+C_2\big),
\eea
provided that ~$T = R\sqrt{\log{m}}$~ and~$\sqrt{m-1} \ge \wg(\TcK) + c.$
\end{thm}

Compared to Theorem \ref{thm:one_main}, the theorem above specifies the quantization parameter $T$ with  an exact constant for Gaussian measurements. In our simulations, we observe that the same value yields good results even for other sub-gaussian distributions. Also, observe the close agreement in the error formulas of the two theorems. The formula  in Theorem \ref{thm:one_main}  has an extra $\sqrt{\log{n}}$ factor; however, note that this term disappears in the overdetermined  regime $m>n$. Finally, in the Appendix \ref{sec:gaussian} we obtain an even tighter expression for the error bound than the one that appears above. It is also shown that the former simplifies to the latter at the cost of requiring that $m$ is large enough and an extra multiplicative constant $C_1$ (that can be shown to approach 2 for increasing $m$).
%

\paragraph{On the classical statistics regime} The pervasive working assumption in classical statistics is that the number of measurements grows large $m\rightarrow\infty$, while the signal dimension $n$ remains fixed. In that context, signal recovery from dithered one-bit measurements using least-squares (i.e., \eqref{eq:muLASSO} without any constraints) was previously studied by Dabeer and Karnik \cite[Thm.~1]{dabeer2006signal}. Compared to their result, Theorem \ref{thm:one_gauss} is valid more generally: (i) it is non-asymptotic; (ii) it applies to the the high-dimensional regime (both $m$ and $n$ large); (iii) it captures the role of the constraint set $\Kc$. As expected, when specialized to the classical statistics regime, our theorem  is in full agreement with \cite[Thm.~1]{dabeer2006signal}, which also captures the $\sqrt{\frac{\log m}{m}}$-rate.

\paragraph{On the dithering distribution} Theorem \ref{thm:one_main} assumes iid uniform dithering signals $\tau_i$ in the interval $[-T,T]$, as well as a specific choice of $T=C_3LR\sqrt{\log{m}}$. A first interesting question is what is the best possible rate as a function of $m$ for uniform dithering signals: is it possible to improve upon the logarithmic rate loss of Theorem \ref{thm:one_main}? A second question asks whether better rates can be obtained with different distributions for the dithering signal. Interestingly, Dabeer and Karnik \cite[Thm.~2]{dabeer2006signal} give a negative answer to the latter question in the classical asymptotic regime with large $m$ and fixed $n$. It is interesting to study these questions in the high-dimensional regime that is of modern interest. 

\paragraph{On the scaling with $R$} Note that the proposed scheme requires knowledge of an upper bound $R$ on the true norm $\|\x_0\|_2$ of the unknown signal. Indeed, both the quantization-scheme parameter $T$ and the G-Lasso parameter $\mu_\Qc$ are required by Theorem \ref{thm:one_main} to be proportional to such a value $R$. Moreover, the theorem predicts that the normalized estimation error $\frac{\|\hat\x-\x_0\|_2}{\|\x_0\|_2}$ scales linearly with the the overshoot $\frac{R}{\|\x_0\|_2}$ in our guess regarding the signal's norm.

\paragraph{Examples} Here, we specialize our results to the two popular instances that were also considered in Section~\ref{sec:uni}.
\begin{itemize}
\item \emph{Sparse recovery}:
 By direct application of \eqref{eq:gw_sparse} in Theorem \ref{thm:one_main} we get that
$$
\|\widehat{\x}-\x_0\|_2 \leq L\,R\cdot\Oc\Big(\frac{\sqrt{s\, \log({n}/{s})}}{\sqrt{m}}\big(\sqrt{\log m} \maxx \sqrt{\log n}  \,\big)\Big)
$$

\item \emph{Low-rand recovery}: 
By direct application of \eqref{eq:gw_lowrank}  in Theorem \ref{thm:one_main} we get that
$$
\|\widehat{\x}-\x_0\|_2 \leq L\,R\cdot\Oc\Big( \frac{r\, n}{\sqrt{m}}\,\big(\sqrt{\log m} \maxx \sqrt{\log n}\,\big)\Big).
$$

\end{itemize}

\paragraph{Related results} As mentioned before, our work is in-part motivated by recent results in \cite{dirksen2018robust}. In the context of uniformly dithered one-bit quantization, Dirksen and Mendelson \cite{dirksen2018robust} propose and analyze a different convex-optimization based estimator which shares some similarity with the LASSO in \eqref{eq:muLASSO}. Specifically, put in our notation, they solve the following program
\bea\label{eq:Mendelson}
\widehat{\x}:=\arg\max_{\x\in\Kc} \frac{1}{m}\sum_{i=1}^m y_i\ab_i^T\x - \frac{1}{2\lambda}\|\x\|_2^2,
\eea
where the value of the regularizer $\lambda$ is set to $\lambda=T$. To see that this is indeed rather similar to the G-Lasso objective in \eqref{eq:muLASSO}, expand the squares in the latter and recall that we set $\mu_{\Qc}=T$. Empirically, we have observed that the two algorithms perform similarly for iid sub-gaussian measurements, when the value of $T$ is set according to our Theorem \ref{thm:uni_main}. However, note that the G-LASSO also works for the uniform quantization scheme with only a simple tuning of the parameter $\mu_\Qc$.   In terms of theoretical results, the error guarantees of Theorem \ref{thm:uni_main} and our suggested value for $T$ are not directly comparable to corresponding results in \cite[Thm.~1.3]{dirksen2018robust}, which are (perhaps) less explicit in terms of the problem parameters, e.g., $m$, $n$, $R$, etc.. For example, our result naturally suggests a ``good" value of $T\propto R\sqrt{\log{m}}$. On the other hand, we mention that \cite[Thm.~1.3]{dirksen2018robust} holds under a more general setting since it also accounts for pre- and post-quantization noise and also yields uniform guarantees over all $\x_0\in\Kc$. We leave such extensions of our results as future work.

\section{Simulations}
\label{sec:simulations}

In this section we experimentally evaluate the recovery performance of the Generalized LASSO in  \eqref{eq:muLASSO} when dithered quantized measurements are available. We consider both the uniform and the one-bit quantization schemes in the presence of uniformly distributed dithering, as shown in \eqref{eq:uni_model} and \eqref{eq:one_model}, respectively. In addition, we compare the performance of the  method with corresponding ones proposed recently in \cite{xu2018taking} and \cite{dirksen2018robust} for uniform and one-bit quantization, respectively.  

We present results in which the measurements vectors have entries that are iid Rademacher random variables\footnote{A Rademacher random variable takes two valued $\pm1$ with equal probability.}.
%
%
 Throughout our experiments, we keep the dimension of the signal fixed $n = 100$. The unknown signal $\x_0$ is chosen to be $s$-sparse constructed as follows. First, we select the support of $\x_0$ uniformly at random among all possible supports. Then, the non-zero entries of $\x_0$ are sampled iid from the standard normal distribution. Finally, we scale the entries of the signal such that  $\|\x_0\|_2 = 8$. For one-bit measurements, we choose $R=10>\|\x_0\|_2$. 
%
In order to estimate $\x_0$, we solve the G-Lasso in \eqref{eq:muLASSO} with $\Kc=\{\x\in\R^n~|~\|\x\|_1\leq \|\x_0\|_1\}$ using the CVX package for Matlab \cite{cvx}. The parameter $\mu_\Qc$ in \eqref{eq:muLASSO} is set to $1$ for uniform quantization and to $R\sqrt{\log{m}}$ for one-bit measurements. 
  Throughout this section, each plot is obtained by averaging over $200$ Monte Carlo realizations, with independently sampled measurement vectors $\{\ab_i\}_{i \in [m]}$, signal vector  $\x_0$, and dithering $\{\tau_i\}_{i \in [m]}$ across different trials.

\begin{figure}[t]

\begin{minipage}[b]{1.0\linewidth}
  \centering
  \centerline{\includegraphics[width=0.75\linewidth]{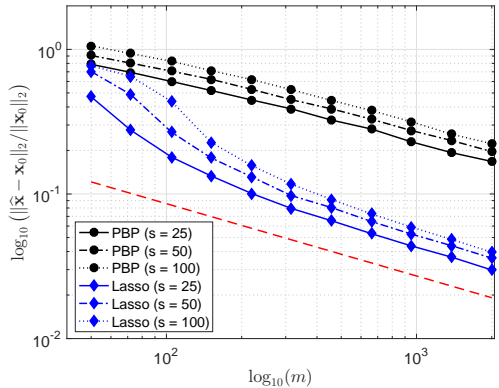}}
\end{minipage}
\caption{Comparison between the G-Lasso (see~\eqref{eq:muLASSO}) and the PBP~\cite{xu2018taking} for  dithered uniform-quantized measurements  (see~\eqref{eq:uni_model}) and for iid Rademacher measurements vectors. The error plots correspond to the following choice of parameters: $\Delta = 3$, $n = 100$, and sparsity $s \in \{25, 50, 100\}$. The red dashed line highlights the $\frac{1}{\sqrt{m}}$ scaling predicted by Theorem \ref{thm:uni_main}.}
\label{fig:uniform-gauss}
\end{figure}

\begin{figure}[t]

\begin{minipage}[b]{1.0\linewidth}
  \centering
  \centerline{\includegraphics[width=0.75\linewidth]{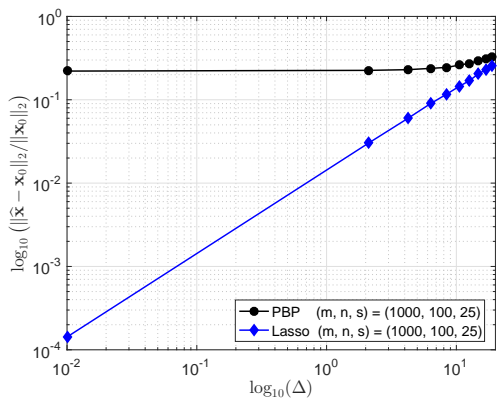}}
\end{minipage}
\caption{Illustration of the error dependence on $\Delta$ for the G-Lasso \eqref{eq:muLASSO} and for the PBP~\cite{xu2018taking}. The plots correspond to the uniform dithered quantizations with iid Rademacher measurement vectors and $(m, n, s) = (1000, 100, 25)$.}
\label{fig:uniform-delta-gauss}
\end{figure}

\begin{figure}[t]

\begin{minipage}[b]{1.0\linewidth}
  \centering
  \centerline{\includegraphics[width=0.8\linewidth]{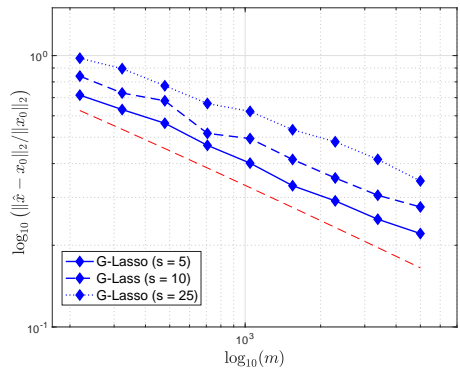}}
\end{minipage}
\caption{Illustration of the recovery performance of the G-Lasso \eqref{eq:muLASSO} for the one-bit dithered quantization (see~\eqref{eq:one_model}) and iid Rademacher measurement vectors. For the error plots we take $n = 100$ and sparsity $s \in \{5, 10, 25\}$. The red dashed line highlights the $\sqrt{\frac{\log m}{m}}$ scaling predicted by Theorem \ref{thm:one_main}.}
\label{fig:1bit-gauss}
\end{figure}

In Figure~\ref{fig:uniform-gauss}, we compare the G-Lasso (cf.,~\eqref{eq:muLASSO}) with the projected back projection (PBP) method from \cite{xu2018taking} for the case of uniform dithered quantization. Observe that  the error plots for \eqref{eq:muLASSO} concur with the $\frac{1}{\sqrt{m}}$ scaling as predicted by Theorem~\ref{thm:uni_main}. Furthermore, note that \eqref{eq:muLASSO} significantly outperforms the PBP method. Another advantage of \eqref{eq:muLASSO} is its dependence on the parameter $\Delta$: when continuously decreasing the value of $\Delta$  (which results into more informative measurements), this leads to sustained decrease in the recovery error (cf.,~Remark~\ref{rem:uni}). On the other hand, the PBP hits an error floor as the measurements increase, thus failing to utilize the information present in the measurements beyond a certain point. This phenomenon is clearly illustrated in Figure~\ref{fig:uniform-delta-gauss}. 

Focusing on the one-bit dithered quantization scheme (cf.~\eqref{eq:one_model}), we plot the recovery performance of the G-Lasso in Figure~\ref{fig:1bit-gauss}. Note that the recovery performance demonstrates a scaling that behaves as $\frac{\sqrt{\log m}}{\sqrt{m}}$ with increasing $m$, as suggested by Theorem~\ref{thm:one_gauss}.  During our experiments, we also observed that the performance of the G-Lasso is almost identical to that the alternative algorithm in \eqref{eq:Mendelson} from \cite{dirksen2018robust}. However, as discussed in Section~\ref{rem:1bit}, our method offers additional benefits such as: (i) An explicit choice of the parameters $\mu_{\Qc}$ appearing in the objective function. (Note the $\mu_{\Qc}$ is associate with the parameter $T$, which controls the range of the dithering.) (ii) An algorithm that is already well-established and can be commonly used both for one-bit, as well as, uniform quantization schemes.


\section{Future work}
\label{sec:conclude}
There are several directions for future work related to the results presented in this paper. Some of the straightforward ones, include extensions to account for pre- and post-quantization noise, as well as, establishing uniform guarantees over all unknown signals of interest. 
Also, in Section \ref{sec:one} we discussed a number of questions regarding other distributions for the dithering signal, the potential optimality of the error rate of $\sqrt\frac{\log{m}}{m}$ in Theorem \ref{thm:one_main}, etc.

One can also imagine extending our results to general quantization schemes, e.g, general number of quantization levels~\cite{NIPS}. Finally, considering other loss functions in \eqref{eq:muLASSO} (e.g., see \cite{genzel2017high}) and the performance of first-order solvers (e.g., see \cite{oymak2016fast}) are also interesting directions to pursue.

%
%


\bibliographystyle{alpha}
\bibliography{compbib}


 \appendices

\section{Proofs for Section \ref{sec:uni}}\label{sec:proofs_uni}

 Throughout the appendices, we drop the subscript $\Qc$ from the parameter $\mu_\Qc$ and simply write $\mu$, instead. Also, constants denoted by $C_1,C_2,\ldots,c_1,c_2,\ldots$ may change from line to line.

\subsection{Proof of Theorem \ref{thm:uni_main}}\label{sec:proof_uni_main}
We follow the proof strategy as described in Section \ref{sec:sketch}. Specifically, we continue from the key inequality in \eqref{eq:key}; recall the definition of the shorthand notation $\mathrm{LB}$ and $\mathrm{UB}$ for the involved terms. 

First, it is shown in Lemma \ref{lem:lower} that there exist constants $C_1,c>0$ such that with probability $0.995$ it holds that $\mathrm{LB}\geq C_1\,L$ provided \eqref{eq:mgeq} holds. Second, we show in Lemma \ref{lem:upper_uni} that $\E[\mathrm{UB}]\leq C_2\,L\,\Delta\cdot\frac{\wg(\TcK)}{\sqrt{m}}$. Thus, by Markov's inequality: $\mathrm{UB}\leq C_3\,L\,\Delta \cdot\frac{\wg(\TcK)}{\sqrt{m}}$ with probability at least $0.995$.

Therefore, by conditioning on the two aforementioned high-probability events and by using a simple union bound it follows from \eqref{eq:key} that with probability $0.99$:
$$
\|\hat\x-\x_0\|_2 \leq \frac{C_3}{C_1}\,\Delta \cdot\frac{\wg(\TcK)}{\sqrt{m}}.
$$
This completes the proof of the theorem.

\subsection{Lower bound}\label{sec:lower}
The following lemma is essentially a restatement of \cite[Thm.~6.3]{TroppBowling} applied to our setting. Its proof is based on Mendelson's small ball method \cite{mendelson1,mendelson2}.
\begin{lem}[Lower Bound -- Mendelson's small-ball method]\label{lem:lower} Suppose that random vectors $\ab_i\in\R^n, i\in[m],$ satisfy Assumption \ref{ass:sub}. Then, for any $\Tc\subset\R^n$, there exist  constants $C,c_1=c_1(L/\alpha), c_2 = c_2(L/\alpha)>0$ such that with probability at least 0.995 it holds:
\bea
\inf_{\w\in\Tc}\frac{1}{m}\sum_{i\in[m]}(\ab_i^T\w)^2 \geq C\,L, \nn
\eea
provided that $m \ge c_1 \cdot \omega^2(\Tc) + c_2$.
\end{lem}

\begin{proof} The statement directly follows from \cite[Theorem 6.3]{TroppBowling} which shows that, for all $t \geq 0$, the following holds with probability at least $1 - e^{-ct^2}$.
\begin{align}
\inf_{\w \in \Tc} \Big(\frac{1}{m}\sum_{i\in[m]}(\ab_i^T\w)^2\Big)^{1/2} \geq C_1\, \frac{\alpha^3}{L^2} \sqrt{m} - C_2\, L \cdot \omega(\Tc) - \alpha t, \nn
\end{align}
where $C_1,C_2>0$ are absolute constants.
\end{proof}

\subsection{Upper bound}\label{sec:proofs_uni_upper}
In this section we derive an upper bound on the quantity $\E[\mathrm{UB}]$ in the RHS of \eqref{eq:key}. We have already decomposed $\E[\mathrm{UB}]$ in two terms in \eqref{eq:Terms}. Moreover, we have shown in \eqref{eq:uni_zero} that Term II is zero! Recall that this is due to property \eqref{eq:mean_uni} of uniform dithered quantization with uniformly distributed dithering (see also Lemma \ref{lem:key_uniform}). Hence, in the remainder we focus on obtaining an upper bound on Term I.

Our main result is summarized in the following lemma.
\begin{lem}[Upper bound -- Uniform case]\label{lem:upper_uni} Let $\ab_i, y_i,~i\in[m]$, and $\x_0$ be as in Theorem \ref{thm:uni_main}. Then, for any subset $\Tc\subset\R^n$, there exists absolute constant $C>0$ such that
\bea
\label{eq:uniform_upp2}
\mathbb{E} \sup_{\w\in\Tc}\frac{1}{m} \sum_{i = 1}^{m}\big(y_i - \ab_i^T\x_0\big)\ab_i^T\w \leq C\,\Delta\,L\cdot\frac{\wg(\Tc)}{\sqrt{m}}.
\eea
\end{lem}

\begin{proof}
We begin with a standard symmetrization trick \cite[Lem.~6.3]{ledoux} introducing iid Rademacher random variables $\varepsilon_i,~i\in[m]$:
\bea\label{eq:sym_uni}
&\Exp\sup_{\w\in\Tc} \frac{1}{m} \sum_{i=1}^{m} (e_i\mathbf{\ab_i}^T\x_0  - \Exp[ \tilde{e}_i{\tilde{\ab}_i}^T\w ] ) \leq \nonumber \\
&~~~~~~~~~~~~~~~~~~~~~~2\cdot\mathbb{E} \sup_{\w\in\Tc} \frac{1}{m}\sum_{i = 1}^{m} \varepsilon_i  e_i\ab_i^T\w.
\eea
where, recall that $e_i,~i\in[m]$ denotes the quantization noise as in \eqref{eq:q_error}.

Next, we observe that the quantization noise $e_i,~i\in[m]$ is always bounded, i.e., $|e_i|\leq \Delta$. We can exploit this and apply contraction principle to further simplify the expression in the RHS of \eqref{eq:sym_uni}. Specifically, we apply Talagrand's Rademacher contraction principle (see \cite[Eqn.~4.20]{ledoux}; also given as Proposition \ref{thm:contraction} for convenience) with 
$
\phi_i(x) = e_i x,~~i \in [m]
$
and $\Sc = \{\mathbf{t}~:~t_i = \ab_i^T\w~\forall i \in [m]~\text{and}~\w \in \Tc \}$. Note that $\phi_i(0) = 0$ and 
\begin{align}
|\phi_i(x) - \phi_i(x')| &\leq |e_i| \cdot |x - x'| \nonumber \leq \Delta \cdot |x - x'|.
\end{align}
Therefore, we have,
\bea
\label{eq:uniform_upp1}
\mathbb{E} \sup_{\w\in\Tc} \sum_{i = 1}^{m} \varepsilon_i  \, e_i\ab_i^T\w \leq \Delta \cdot \mathbb{E} \sup_{\w\in\Tc} \sum_{i = 1}^{m} \varepsilon_i  \,\ab_i^T\w. 
\eea
 
 Now, we can directly relate the expected supremum on the RHS above with the Gaussian width of the set $\Tc$ thanks to Talagrand's majorizing theorem. Specifically, by sub-gaussianity of the $\ab_i$'s, the random vector $\h:=\sum_{i = 1}^{m} \varepsilon_i  \,\ab_i$ is also sub-gaussian and satisfies $\psin{\h}\leq C\,L\,\sqrt{m}$ (e.g., \cite[Prop.~2.6.1]{VerBook}). Therefore, because of generic chaining \cite[Thm.~1.2.6]{Tal} and the majorizing measure theorem \cite[Thm. 2.1.1]{Tal}:
\bea\label{eq:tal_uni}
\frac{1}{\sqrt{m}}\E \sup_{\w\in\Tc}\h^T\w\leq C\,L\cdot\wg(\Tc).
\eea

Combining \eqref{eq:sym_uni}, \eqref{eq:uniform_upp1}, \eqref{eq:tal_uni}, and the fact that $\E[\mathrm{Term~II}]=0$ we conclude with the desired inequality in \eqref{eq:uniform_upp2}.

\end{proof}

\section{Proofs for Section \ref{sec:one}}\label{sec:proofs_one}

\subsection{Proof of Theorem \ref{thm:one_main}}
We follow the proof strategy as described in Section \ref{sec:sketch}. Specifically, we continue from the key inequality in \eqref{eq:key}; recall the definition of the shorthand notation $\mathrm{LB}$ and $\mathrm{UB}$ for the involved terms. 

First, note that the term $\mathrm{LB}$ does \emph{not} depend on the quantization scheme. Thus, we can use the result of Section \ref{sec:lower}. In particular, it is shown in Lemma \ref{lem:lower} that there exists constants $C_1,c>0$ such that with probability $0.995$ it holds that $\mathrm{LB}\geq C_1\,L$ provided that \eqref{eq:mgeq} holds.

Second, combining the results of Lemmas \ref{lem:upper_two} and \ref{lem:upper_one} we show that 
\bea
&\E[\mathrm{UB}] =\mathrm{Term~I} + \mathrm{Term~II} \nonumber \\
&\leq C_1L^2R\sqrt{\frac{\log{(mn)}}{m}}\cdot\big(\wg(\TcK) + C_2 \big) + C_3L^2R\sqrt{\frac{\log{m}}{m}}\nn\\
&\leq C_4L^2R\sqrt{\frac{\log{(mn)}}{m}}\cdot\big(\wg(\TcK) + C_5 \big).
\eea
Thus, by Markov's inequality the same bound holds for the random variable $\mathrm{UB}$ with probability at least $0.995$.

We can conclude the proof of the theorem by repeating mutatis mutandis the last argument in the proof of Theorem \ref{thm:uni_main} in Section \ref{sec:proof_uni_main}.

\subsection{Upper bound}
In this section we upper bound the two terms in \eqref{eq:Terms} under the setting of Theorem \ref{thm:one_main}.

\subsubsection{Upper Bounding Term II}
The following lemma is only a slight modification of \cite[Cor.~5.2]{dirksen2018robust}
\begin{lem}[Term II -- one-bit case]\label{lem:upper_two}
Suppose that the vector $\ab\in\R^n$ satisfies Assumption \ref{ass:sub} and let $\Tc\subset\Sc^{n-1}$ be an arbitrary subset of the unit sphere. Furthermore, let the measurement $y$ be given as in \eqref{eq:uni_model}. Recall that $\|\x_0\|_2\leq R$. Then, there exist absolute constants $c,C>0$ such that if $T= cLR\sqrt{\log{m}}$, it holds that
\bea
\sup_{\w \in \Tc} \mathbb{E}\big[(T\cdot y  - \ab^T\x_0 ){\ab}_i^T\w\big] \leq CL^2R \sqrt{\frac{{\log{m}}}{{m}}}.
\eea
\end{lem}

\begin{proof} 
For simplicity, let us call $g:=\ab^T\x_0$ and $h:=\ab^T\w$.
Recall from \eqref{eq:mean_one} that
\bea\nn
\E_{\tau} [T \sign(g+\tau) - g] &= -g\ind{|g|>T} + T\ind{g>T} \nonumber \\
&~~~~~-T\ind{g<-T}.
\eea
Using this along with the tower property of expectation and Cauchy-Schwartz inequality, we obtain that
\bea
&\mathbb{E}\big[(T\cdot y  - \ab^T\x_0 ){\ab}^T\w\big] \nonumber \\ 
&~~~~~~~= \E\big[ -gh\ind{|g|>T} + hT\ind{g>T} - hT\ind{g<-T} \big] 
\nn\\ 
&~~~~~~~\leq \sqrt{\E[h^2]} \Big( \sqrt{\E[g^2\ind{|g|>T}]} + \sqrt{\E[T^2\ind{g>T}]}  \nonumber \\
&~~~~~~~~~~~~~~~~~~~~~~~~~~~~+\sqrt{\E[T^2\ind{g<-T}]} ~\Big) \nn\\
&~~~~~~~\leq 3\cdot\sqrt{\E\big[h^2\big]}\cdot\sqrt{\E\big[g^2\ind{|g|>T}\big]}.\label{eq:prev_234}
\eea
By sub-gaussian properties of $\ab$ (see Lemma \ref{lem:sub_marginals}) $(\E[h^2)])^{1/2}\leq L \|\w\|_2$. Also, using integration by parts and subgaussian tails of $\ab^T\x_0$ (again, see Lemma \ref{lem:sub_marginals})  it can be shown (as in \cite[Lem.~5.1]{dirksen2018robust}) that for some absolute constant $c_1>0$ it holds
\bea\label{eq:IBP}
\E g^2\ind{|g|>T} &\leq  T^2\Pro(|g|>T)+ 2\int_{T}^{\infty}t\,\Pro(|g|>t)\mathrm{d}t\nn\\ &
\leq \big(2T^2 + \frac{L^2\cdot\|\x_0\|^2}{c_1}\big)\cdot e^{-\frac{c_1 T^2}{L^2\|\x_0\|^2}}.
\eea
Thus, continuing from \eqref{eq:prev_234} and using the fact that $\Tc\subset\Sc^{n-1}\Rightarrow\sup_{\w\in\Tc}\|\w\|_2=1$, we get
\bea
&\sup_{\w\in\Tc}\mathbb{E}\big[(T\cdot y  - \ab^T\x_0 ){\ab}^T\w\big] \leq \nonumber \\ 
&~~~~~~~~~~~~~~~~~~~~~~~~C_1 L\sqrt{2c_1T^2 + L^2\|\x_0\|^2} \cdot e^{- \frac{c_1T^2}{L^2 \|\x_0\|^2}}.
\eea
To complete the proof, set $$T=\frac{1}{\sqrt{2c_1}}LR\sqrt{\log{m}}$$ and use the fact that $\|\x_0\|_2\leq R$ to find that the exponential term above is upper bounded by $1/\sqrt{m}$ and the rest by $C_1 L^2R(2\sqrt{\log{m}}+1)\leq C L^2R \sqrt{\log{m}}$ for  constants $C_1,C>0$ (we trivially assume that $m\geq2$).
\end{proof}

\subsubsection{Upper Bounding Term I}
The next lemma establishes an upper bound on Term I.
\begin{lem}[Term I -- One-bit case]\label{lem:upper_one} Let $\ab_i, y_i,~i\in[m]$, $\x_0$, and $R$ be as in Theorem \ref{thm:one_main}. Let $e_i,~i\in[m]$, denote the quantization noise as in \eqref{eq:q_error}. Finally, suppose that $\mu= cLR\sqrt{\log{m}}$ for some absolute constant $c>0$. Then, for any subset of the unit sphere $\Tc\subset\Sc^{n-1}$ there exist absolute constants $C_1,C_2>0$ such that it holds that
\bea
\label{eq:one_up1}
&\Exp\sup_{\w\in\Tc} \frac{1}{m} \sum_{i=1}^{m} (e_i\mathbf{\ab_i}^T\x_0  - \Exp[ \tilde{e}_i{\tilde{\ab}_i}^T\w ] ) \leq \nonumber \\
&~~~~~~~C_1 L^2 R \sqrt{\frac{\log{(mn)}}{m}}\cdot\big(\wg(\Tc) + C_2 \big).
\eea
\end{lem}


\begin{proof}

Exactly as in the proof of Lemma \ref{lem:upper_uni} for the uniform quantization case, we begin with a standard symmetrization trick \cite[Lem.~6.3]{ledoux} introducing iid Rademacher random variables $\varepsilon_i,~i\in[m]$:
\bea\label{eq:sym_one}
&\Exp\sup_{\w\in\Tc} \frac{1}{m} \sum_{i=1}^{m} (e_i\mathbf{\ab_i}^T\x_0  - \Exp[ \tilde{e}_i{\tilde{\ab}_i}^T\w ] ) \leq \nonumber \\
&~~~~~~~~~~~~~~~~~2\cdot\mathbb{E} \sup_{\w\in\Tc} \frac{1}{m}\sum_{i = 1}^{m} \varepsilon_i  e_i\ab_i^T\w.
\eea
Recall that in Lemma \ref{lem:upper_uni} we proceeded by using the fact that in uniform quantization scheme the quantization error $e_i:=\mu\cdot \sign(\ab_i^T\x_0+\tau_i)-\ab_i^T\x_0$ is always a bounded random variable. This allowed us to use the Rademacher contraction principle. Unfortunately, the $e_i$'s are not bounded in one-bit quantization. However, as we will see they can be bounded by a sufficiently large threshold with high-probability. Towards that goal we introduce indicator random variables as follows:
\bea\label{eq:ind_delta}
\delta_i:=\ind{|\ab_i^T\x_0|\leq\nu},\quad \deltac_i:=1-\delta_i,
\eea
where the value of $\nu>0$ is to be specified later in the proof. With these and using the triangle inequality for the supremum metric we write
\bea
&\E \sup_{\w\in\Tc} \frac{1}{m}\sum_{i = 1}^{m} \varepsilon_i  e_i\ab_i^T\w \leq \nn\\
&~~~~\underbrace{\E \sup_{\w\in\Tc} \frac{1}{m}\sum_{i = 1}^{m} \varepsilon_i  \delta_ie_i\ab_i^T\w}_{\mathrm{Term~A}} + \underbrace{\E \sup_{\w\in\Tc} \frac{1}{m}\sum_{i = 1}^{m} \varepsilon_i  \deltac_ie_i\ab_i^T\w}_{\mathrm{Term~B}}.\label{eq:TermsAB}
\eea
We proceed by bounding the two terms above.

\vspace{2pt}
\underline{$\mathrm{Term~A}$}: Note that $|\delta_ie_i|\leq \mu+\nu,~i\in[m].$
Hence, we may bound Term A by first applying Rademacher contraction principle. The details are exactly identical to what follows Eqn. \eqref{eq:sym_uni} in the proof of Lemma \ref{lem:upper_uni}, thus, they are omitted from brevity. Repeating the steps as in \eqref{eq:uniform_upp1} and \eqref{eq:tal_uni} (replacing $\Delta$ wit $\mu+\nu$), we conclude that
\bea
\mathrm{Term~A}&\leq (\mu+\nu)\cdot\Exp\sup_{\w\in\Tc}\frac{1}{m}\sum_{i\in[m]}\varepsilon_i\ab_i^T\w\nn \\  
&\leq C_1\,(\mu+\nu)\,L\cdot\frac{\wg(\Tc)}{\sqrt{m}},\label{eq:TermA}
\eea
for appropriate absolute constant $C_1>0$.

\vspace{2pt}
\underline{$\mathrm{Term~B}$}: We use the following crude bound on the supremum (recall that $\Tc\subset\Sc^{n-1}$) to obtain the following chain of inequalities:
\bea
\mathrm{Term~B}&\leq \E\Big\|\frac{1}{m}\sum_{i = 1}^{m} \varepsilon_i  \deltac_ie_i\ab_i\Big\|_2 \overset{(i)}{\leq} \E\Big[\frac{1}{m}\sum_{i = 1}^{m} \Big\| \varepsilon_i  \deltac_ie_i\ab_i\Big\|_2\Big]\nn\\
&\overset{(ii)}{\leq} \E\big\|\varepsilon_1  \deltac_1 e_1 \ab_1\big\|_2 = \E\big[\, |e_1\deltac_1|\cdot\|\ab_1\|_2 \big]\nn\\
&=\E\big[\, |\deltac_1 (\mu\cdot\sign(g_1+\tau_1)-g_1)|\cdot\|\ab_1\|_2 \big]\nn\\
&\overset{(iii)}{\leq} \mu\cdot\E\big[\, |\deltac_1|\cdot\|\ab_1\|_2 \big] +\E\big[\, |\deltac_1g_1|\cdot\|\ab_1\|_2 \big] \nn\\
&\overset{(iv)}{\leq} \mu\cdot\sqrt{\E\deltac_1}\sqrt{\E\|\ab_1\|_2^2} + \sqrt{\E\deltac_1g_1^2}\sqrt{\E\|\ab_1\|_2^2}\nn\\
&\overset{(v)}{\leq}\ C_2 L (\mu\cdot\sqrt{\E\deltac_1} + \sqrt{\E\deltac_1g_1^2})\cdot \sqrt{n},\label{eq:TermB_1}
\eea
where, we have denoted 
$$
g_1:=\ab_1^T\x_0, 
$$
and,
$(i)$ and $(iii)$ follow from the triangle inequality; $(ii)$ follows by combining the linearity of expectation with the fact that $\varepsilon_i  \deltac_ie_i\ab_i,~i\in[m]$, are identically distributed; $(iv)$ follows from the Cauchy--Schwarz inequality; $(v)$ follows from Lemma \ref{lem:norm_sub}.
 
Continuing, note that  $g_1$ is sub-gaussian with $ 
\|g_1\|_{\psi_2} \leq L\cdot \|\x_0\|_2.$ Therefore,
\bea\label{eq:E1}
\E\deltac_1 = \Pro(|g_1|>\nu)\leq 2e^{-c_1\frac{\nu^2}{L^2\|\x_0\|_2^2}},
\eea
and using integration by parts exactly as in \eqref{eq:IBP}:
\bea\label{eq:E2}
\E g_1^2\deltac_1 
\leq \big(2\nu^2 + \frac{L^2\cdot\|\x_0\|^2}{c_1}\big)\cdot e^{-\frac{c_1 \nu^2}{L^2\|\x_0\|^2}}.
\eea
At this point, choose 
\bea\label{eq:nu}
\nu = \frac{1}{\sqrt{c_1}}LR\sqrt{\log{(mn)}}.
\eea
With that choice, we deduce from \eqref{eq:E1} and \eqref{eq:E2} that
\bea
\E\deltac_1 &\leq \frac{2}{mn},\nn\\
\E g_1^2\deltac_1 &\leq \frac{1}{c_1}L^2R^2\frac{(2\log{(mn)}+1)}{mn}\nn.
\eea
By putting these together in \eqref{eq:TermB_1} and trivially assuming that $mn \geq 2$, we conclude that
\bea
\mathrm{Term~B}\leq C_3L\frac{\mu}{\sqrt{m}} + C_4L^2R\frac{\sqrt{\log{(mn)}}}{\sqrt{m}}.\label{eq:TermB}
\eea

\vspace{2pt} We are now ready to finish the proof of the lemma. Recall the value of $\mu$ in the statement of the lemma and \eqref{eq:nu}. Note that $\mu\leq c'\nu$ for some constant $c'>0$. Hence, putting together \eqref{eq:TermA} and \eqref{eq:TermB} in \eqref{eq:TermsAB} we find that
\bea 
&\E \sup_{\w\in\Tc} \frac{1}{m}\sum_{i = 1}^{m} \varepsilon_i  e_i\ab_i^T\w \leq \nn\\
&\leq C_5 L^2 R \sqrt{\frac{\log{(mn)}}{m}}\,\wg(\Tc) + C_6 L^2R\sqrt{\frac{\log{(mn)}}{m}}\nn\\
&\leq C_7 L^2 R \sqrt{\frac{\log{(mn)}}{m}}\big(\wg(\Tc) + C_8 \big)\nn.
\eea
In view of \eqref{eq:sym_one}, this completes the proof.
\end{proof}

\subsection{Proof of Theorem \ref{thm:one_gauss}}\label{sec:gaussian}

We continue the proof from \eqref{eq:key}. The lower bound follows directly from Gordon's escape through a mesh theorem \cite{gorLem}. Since this is classical, we omit the details for brevity; see for example \cite{plan2016generalized}. Onwards, we focus on the upper-bound term $\mathrm{UB}$. The proof follows the lines of th proof of \cite[Lemma~4.3]{plan2016generalized}, but requires several modifications.

For $i \in [m]$, we decompose $\ab_i$ into two components along the direction $\x_0$ and the space perpendicular to $\x_0$, repsectively, i.e.,
\begin{align}
\label{eq:decom}
\ab_i^T = \ab_i^TP  + \ab_i^TP^{\perp} =  \frac{\ab_i^T\x_0}{\|\x_0\|_2}\cdot \frac{\x^T_0}{\|\x_0\|_2} + \ab_i^TP^{\perp}.
\end{align}
Thus, $\mathrm{UB}$ can be decomposed in the following two terms, which we bound separately.
\begin{align}
\label{eq:terms}
&\underbrace{\frac{1}{m} \sum_{i = 1}^{m}\big(\mu\cdot y_i - \ab_i^T\x_0\big) \frac{\ab_i^T\x_0}{\|\x_0\|_2}\cdot \frac{\x^T_0\w}{\|\x_0\|_2}}_{\rm Term~I} + \nn \\
&~~~~~~~~~~~~~~~~\underbrace{\frac{1}{m} \sum_{i = 1}^{m}\big(\mu\cdot y_i - \ab_i^T\x_0\big)\ab_i^TP^{\perp}\w}_{\rm Term~II}.
\end{align}

\vspace{2pt}
\underline{\rm Term~I:}~Since $\TcK\subset\Sc^{n-1}$, note that 
\begin{align}
\mathbb{E}\sup_{\w \in\TcK} {\rm Term~I}&\leq \mathbb{E}\big|\frac{1}{m}\sum_{i\in [m]}\ksi_i\big| = \mathbb{E}\Big({\big|\frac{1}{m}\sum_{i\in [m]}\ksi_i\big|^2}\Big)^{1/2}\nn\\
&\leq\Big(\mathbb{E}\big|\frac{1}{m}\sum_{i\in [m]}\ksi_i\big|^2\Big)^{1/2}\nn\\
&\leq\frac{1}{\sqrt{m}}\sqrt{\E[\ksi_1^2] + m \E[\ksi_1]^2},\label{eq:TermI_1}
\end{align}
where, we define
\begin{align*}
\zeta_i &:= \frac{\ab_i^T\x_0}{\|\x_0\|_2} \sim \mathcal{N}(0,1), \\
\xi_i &:= \big(\mu\cdot \sign(\|\x_0\|_2\zeta_i + \tau_i) - \|\x_0\|_2\cdot\zeta_i \big)\cdot \zeta_i,
\end{align*}
and where we have used Jensen's inequality in the second line and the fact that $\ksi_i, i\in[m]$ are iid, in the last line.

Using integration by parts it can be shown that
\begin{align}
\mathbb{E}_{\zeta_i} \zeta_i\cdot \sign(\|\x_0\|_2\zeta_i + \tau_i) 
&= \sqrt{\frac{2}{\pi}}e^{-\frac{\tau_i^2}{2\|\x_0\|_2^2}}\nn.
\end{align}
and
\begin{align}
&\mathbb{E}_{\tau_i,\zeta_i}\zeta_i\cdot\sign(\|\x_0\|_2\cdot\zeta_i + \tau_i) = \nonumber \\
&~~~~~~~~~~~~~~~~~~~~\frac{\|\x_0\|_2}{T}\cdot\Big(1 - 2Q\Big(\frac{T}{\|\x_0\|_2}\Big)\Big).\nn
\end{align}
Combining the above two displays yields:
\begin{align}
\label{eq:xi}
\mathbb{E} \xi_i = \|\x_0\|_2\Big(\frac{\mu}{T} - 1\Big) - 2\frac{\mu}{T}Q\Big( \frac{T}{\|\x_0\|_2}\Big),
\end{align}
and
\begin{align}
\label{eq:xi2}
\mathbb{E}\xi_i^2 &= 3\|\x_0\|_2^2 + \mu^2 - 6\|\x_0\|_2^2\frac{\mu}{T} + \nonumber \\
&~~~~12\|\x_0\|_2^2\frac{\mu}{T}Q\Big(\frac{T}{\|\x_0\|_2}\Big) +  2\mu \sqrt{\frac{2}{\pi}}\|\x_0\|_2e^{-\frac{T^2}{2\|\x_0\|_2^2}}.
\end{align}
In particulare, for $\mu = T = R\sqrt{\log m}$, we have from \eqref{eq:xi} that 
\begin{align}
\label{eq:xi11}
\big|\mathbb{E}\xi_i\big| \leq 2\|\x_0\|_2e^{-\frac{R^2\log m}{2\|\x_0\|_2^2}} = 2\frac{\|\x_0\|_2}{\sqrt{m}} \leq \frac{2R}{\sqrt{m}},
\end{align}
and 
\begin{align}
\mathbb{E}\xi_i^2 &\leq R^2\log m - 3\|\x_0\|_2^2 + 12\|\x_0\|_2e^{-\frac{R^2}{2\|\x_0\|_2^2}\log m} \nonumber \\
&~~~~~~~~~~~+2\sqrt{\frac{2}{\pi}}R\|\x_0\|_2\sqrt{\log m} e^{-\frac{R^2}{2\|\x_0\|_2^2}\log m}\nn \\
&\leq R^2\log m + 12 \frac{\|\x_0\|_2}{\sqrt{m}} + 2\sqrt{\frac{2}{\pi}}R^2\sqrt{\frac{\log m}{m}}.\label{eq:xi21}
\end{align}
Now, we can put  \eqref{eq:xi11} and \eqref{eq:xi21} together in \eqref{eq:TermI_1} to find that
\begin{align}
\label{eq:term1-i}
&\mathbb{E}\sup_{\w \in \TcK} {\rm Term~I} \nonumber \\
&\overset{(i)}{\leq }\frac{1}{\sqrt{m}} \sqrt{ R^2\log m + 12 \frac{R}{\sqrt{m}} + 2\sqrt{\frac{2}{\pi}}R^2\sqrt{\frac{\log m}{m}} + 4R^2}.
\end{align} 

\underline{\rm Term~II:}~Note that $\ab_i^T\x_0$ is independent of $\ab_i^TP^\perp\w$ for all $\w$ and all $i\in [m]$ (see also \cite[Lemma 4.3]{plan2016generalized}). Hence,
\begin{align}
\label{eq:term2-i}
&\mathbb{E}\sup_{\w \in \TcK} {\rm Term~II} = \nonumber \\
&~~~\mathbb{E}\sup_{\w \in \TcK}\frac{1}{m} \sum_{i = 1}^{m}\big(\mu\cdot\sign(\widetilde{\ab}_i^T\x_0 + \tau_i) - \widetilde{\ab}_i^T\x_0\big)\ab_i^T\w,
\end{align}
where $\{\widetilde{\ab}_i\}_{i \in [m]}$ are iid random vectors with distribution $\mathcal{N}(0, \Id)$ and, most importantly, independent of the measurement vectors $\{\ab_i\}_{i \in [m]}$. Denoting $$\eta_i := \mu \cdot \sign(\widetilde{\ab}_i^T\x_0+\tau_i) - \widetilde{\ab}_i^T\x_0,$$ it follows from \eqref{eq:term2-i} that 
\begin{align}
\label{eq:term2-ii}
&\mathbb{E}\sup_{\w \in\TcK} {\rm Term~II} = \frac{1}{m} \mathbb{E} \sup_{\w \in \TcK} \sum_{i = 1}^{m} \eta_i \ab_i^T\w \nonumber \\
&~~~~~~~\overset{(i)}{=}\frac{1}{m}\mathbb{E}\Big({\sum_{i=1}^{m}\eta_i^2}\Big)^{1/2} \cdot \wg(\TcK) \leq \big({\mathbb{E}\eta_i^2}\big)^{1/2} \cdot \frac{\wg(\TcK)}{\sqrt{m}},
\end{align}
where $(i)$ follows from the fact that for fixed $(\eta_1,\ldots, \eta_m)$, $\sum_{i = 1}^{m}\eta_i \ab_i$ is random variable with the distribution $\mathcal{N}(0, \sum_{i=1}^{m}\eta_i^2)$. Note that
\begin{align}
\label{eq:Eeta}
\mathbb{E}\eta_1^2 &= \mathbb{E} \big(\mu\cdot\sign(\|\x_0\|_2 \zeta_1 + \tau) - \|\x_0\|_2\zeta_1\big)^2 \nonumber \\
&= \mu^2 + \|\x_0\|_2^2 - 2\frac{\mu}{T}\|\x_0\|_2^2\big(1 - 2Q\big(\frac{T}{\|\x_0\|_2}\big)\big) \nonumber \\
&\overset{(i)}{=}  R^2\cdot \log m + \|\x_0\|_2^2 - 2\|\x_0\|_2^2 + \nonumber \\
&~~~~~~4 Q\Big(\frac{R}{\|\x_0\|_2}\sqrt{\log m}\Big) \cdot \|\x_0\|_2^2 \nonumber \\
&\leq R^2 \cdot \log m + 4R^2 \cdot\frac{1}{\sqrt{m}},
\end{align}
where we have used $\mu = T = R\sqrt{\log m}$ in $(i)$. By combining \eqref{eq:term2-ii} and \eqref{eq:Eeta}, we obtain that 
\begin{align}
\label{eq:term2-iii}
\mathbb{E}\sup_{\w \in \mathcal{K} \cap \mathcal{S}} {\rm Term~II} &\leq 2R\sqrt{\log m + \frac{4}{\sqrt{m}}} \cdot \frac{\wg(\TcK)}{\sqrt{m}}. 
\end{align}

\vspace{2pt}
To continue, we only need to combine \eqref{eq:term1-i} and \eqref{eq:term2-iii}. In fact, letting $m$ large enough, we can find constants $C_1,C_2>0$ such that 
\bea\nn
\eqref{eq:term1-i}\leq C_1R\sqrt{\frac{\log{m}}{m}},~~\eqref{eq:term2-iii}\leq C_2R\sqrt{\frac{\log{m}}{m}}\cdot\wg(\TcK),~~
\eea
Thus,
$$
\E[\mathrm{UB}]\leq C_3 R\sqrt{\frac{\log{m}}{m}}(\wg(\TcK)+C_4).
$$

The proof is now complete by applying Markov's inequality and combining with the lower bound, just as in the proof of Theorem \ref{thm:one_main}.

\section{Auxiliary Facts}

In this section, we gather a few auxiliary results that are used in the proofs in Sections \ref{sec:proofs_uni} and \ref{sec:proofs_one}.

The following result is classical in the theory of dithered quantization; see \cite{gray1993dithered} and references therein. We include here a proof for completeness.
\begin{lem}[Quantization error -- Uniform Quantization]\label{lem:key_uniform}
Let $\tau$ be a random variable distributed according to ${\rm Unif}\big(-\frac{\Delta}{2}, \frac{\Delta}{2}\big]$. Then for a fixed $x \in \R$, we have 
\begin{align}
\label{eq:xtau}
\mathbb{E}_{\tau}\Delta\Big(\Big\lfloor\frac{x + \tau}{\Delta}\Big\rfloor + \frac{1}{2} \Big) = x.
\end{align}
\end{lem}
\begin{proof}
Without loss of generality, we assume that $\Delta = 1$. 
Note that depending on whether $0 \leq x - \floor{x} < \frac{1}{2}$ or $\frac{1}{2} \leq x - \floor{x} \leq 1$, we have
\begin{align}
\label{eq:xtau_1}
\floor{x + \tau} = \begin{cases}
\floor{x} - 1 & \text{if}~\tau \in \big(-\frac{1}{2}, \floor{x} - x\big),\\
\floor{x}  & \text{otherwise.}
\end{cases}
\end{align}
or 
\begin{align}
\label{eq:xtau_2}
\floor{x + \tau} = \begin{cases}
\floor{x} + 1 & \text{if}~\tau \in \big(\floor{x} +1 - x, \frac{1}{2}\big],\\
\floor{x}  & \text{otherwise.}
\end{cases}
\end{align}
respectively. We establish \eqref{eq:xtau} assuming that  $0 \leq x - \floor{x} < \frac{1}{2}$. The proof for the remaining case follows using the similar steps. Note that
\begin{align}
\mathbb{E}_{\tau}\Big(\floor{x + \tau}+ \frac{1}{2}\Big) &= \int_{-\frac{1}{2}}^{\frac{1}{2}} \floor{x + \tau} d\tau + \frac{1}{2} \nonumber \\ 
& \overset{(i)}{=}  \int_{-\frac{1}{2}}^{\frac{1}{2}} \floor{x} d\tau - \int_{-\frac{1}{2}}^{\floor{x} - x} d\tau  + \frac{1}{2} \nonumber \\
& = \floor{x} - \big(\floor{x} - x + \frac{1}{2}\big) + \frac{1}{2} = x,
\end{align}
where $(i)$ follows from \eqref{eq:xtau_1}.
\end{proof}

The following facts about sub-gaussian random vectors are also well-know. We collect them here for ease of reference, since they are repeatedly used throughout the proofs.

\begin{lem}[Sub-gaussian marginals]\label{lem:sub_marginals}
Let $\ab\in\R^n$ be an $L$-subgaussian random vector. Then, for all $\x\in\R^n$, the random variable $g:=\ab^T\x$ is sub-gaussian. In particular,
\begin{enumerate}
\item $\psin{g}\leq L\|\x\|_2,$

\item $(\E[g^2])^{1/2}\leq CL\|\x\|_2,$ for some universal constant $C>0$.

\item $\Pro(|g|>t)\leq 2e^{-\frac{ct^2}{L^2\|\x\|_2^2}},$ for all $t>0$ and some universal constant $c>0$.
\end{enumerate}
\end{lem}
\begin{proof}
The first statement follows easily by Definition \ref{def:sub}: 
\bea\|g\|_{\psi_2} = \|\x\|_2  \|\ab_i^T\frac{\x}{\|\x\|}\|_{\psi_2} 
\leq \|\x\|_2 \sup_{\|\ub\|_2\leq 1} \psin{\ab_i^T\ub} = \|\x\| L.\nn
\eea
The other two statements are then immediate by the standard equivalent properties of sub-gaussians, e.g.\cite[Sec.~2.5.2]{VerBook}.\end{proof}

\begin{lem}[Norm of sub-gaussian vector]\label{lem:norm_sub}
For an $L$-subgaussian random vector $\ab\in\R^n$ it holds
$$
\E\|\ab\|_2^2 \leq CL^2n,
$$
for some absolute constant $C>0$.
\end{lem}
\begin{proof}
The statement is a result of the following chain of inequalities:
\bea
\E\|\ab\|_2^2 = \E\sum_{i\in[n]}a_i^2 \leq C \sum_{i\in[n]} \psin{a_i}^2 \leq CL^2n,
\eea
where we applied Lemma \ref{lem:sub_marginals} on the entries of $\ab$ denoted as $a_i:=\ab_i^T\e_i,~i\in[n]$, where $\e_i$ is the $i^{\text{th}}$ standard basis vector.
\end{proof}

Finally, throughout our proofs we use the Rademacher comparison principle in the following form.

\begin{propo}[Rademacher comparison principle; Eqn.~(4.20) in \cite{ledoux}]\label{thm:contraction}
Let $f~:~\R \to \R$ be a convex and increasing function. For $i \in [m]$, let $\phi_i : \R \to \R$ be a Lipschitz function with Lipschitz constant $\rho$, i.e., 
$
| \phi_i(x) - \phi_i(x')| \leq \rho |x - x'|,
$
such that $\phi_i(0) = 0$. Then, for any $\Sc \subseteq \R^m$, we have 
\begin{align}
\mathbb{E}f\Big(\sup_{\mathbf{t} = (t_1,\ldots, t_m) \in \Sc} \sum_{i = 1}^{m} \varepsilon_i \phi_i(t_i)\Big)  \leq \mathbb{E} f \Big( \rho \cdot \sup_{\mathbf{t} \in \Sc} \sum_{i = 1}^{m}\varepsilon_i t_i \Big),
\end{align}
where $\{\varepsilon_i\}_{i \in [m]}$ denote $m$ i.i.d. Rademacher random variables.
\end{propo}

\end{document}